\documentclass[11pt,letterpaper]{article}

\usepackage[T1]{fontenc}
\usepackage[cp1251]{inputenc}
\usepackage{textcomp}
\usepackage[centertags]{amsmath}
\usepackage{amsfonts}
\usepackage{amsthm}
\usepackage{amssymb}
\usepackage{enumitem}
\usepackage[numbers,sort&compress]{natbib}
\usepackage[hypertex]{hyperref}
\usepackage{calc}

\usepackage{paperinitial}
\paperinitialization{15mm}{15mm}{20mm}{10mm}{2pt}{10pt}

\DeclareMathOperator{\pr}{pr}

\DeclareMathOperator{\im}{Im}
\DeclareMathOperator{\sgn}{sgn}
\DeclareMathOperator{\Sp}{Sp}
\DeclareMathOperator{\Texp}{Texp}
\DeclareMathOperator{\rot}{rot}
\DeclareMathOperator{\arcth}{arth}

\newcommand{\e}{\varepsilon}
\newcommand{\vf}{\varphi}

\newcommand{\al}{\alpha}
\newcommand{\be}{\beta}
\newcommand{\ga}{\gamma}
\newcommand{\Ga}{\Gamma}
\newcommand{\de}{\delta}
\newcommand{\De}{\Delta}

\newcommand{\la}{\lambda}
\newcommand{\La}{\Lambda}
\newcommand{\ups}{\upsilon}

\newcommand{\spx}{\mathbf{x}}
\newcommand{\spy}{\mathbf{y}}
\newcommand{\spp}{\mathbf{p}}
\newcommand{\spk}{\mathbf{k}}
\newcommand{\spe}{\mathbf{e}}

\newcommand{\ha}{\hat{a}}
\newcommand{\had}{\hat{a}^\dag}
\newcommand{\hb}{\hat{b}}
\newcommand{\hbd}{\hat{b}^\dag}
\newcommand{\hZ}{\hat{Z}}
\newcommand{\hps}{\hat{\psi}}
\newcommand{\hpsd}{\hat{\psi}^\dag}
\newcommand{\tP}{\tilde{P}}
\newcommand{\lan}{\langle}
\newcommand{\ran}{\rangle}

\newtheorem{thm}{Theorem}

\begin{document}
\setlength{\unitlength}{1pt}
\allowdisplaybreaks[4]
\frenchspacing

\title{{\Large\textbf{Inclusive probability of particle creation on classical backgrounds}}}

\date{}

\author{P.O. Kazinski\thanks{E-mail: \texttt{kpo@phys.tsu.ru}}\\[0.5em]
{\normalsize Physics Faculty, Tomsk State University, Tomsk 634050, Russia}}

\maketitle

\begin{abstract}

The quantum theories of boson and fermion fields with quadratic nonstationary Hamiltoanians are rigorously constructed. The representation of the algebra of observables is given by the Hamiltonian diagonalization procedure. The sufficient conditions for the existence of unitary dynamics at finite times are formulated and the explicit formula for the matrix elements of the evolution operator is derived. In particular, this gives the well-defined expression for the one-loop effective action. The ultraviolet and infrared divergencies are regularized by the energy cutoff in the Hamiltonian of the theory. The possible infinite particle production is regulated by the corresponding counterdiabatic terms. The explicit formulas for the average number of particles $N_D$ recorded by the detector and for the probability $w(D)$ to record a particle by the detector are derived. It is proved that these quantities allow for no-regularization limit and, in this limit, $N_D$ is finite and $w(D)\in[0,1)$. As an example, the theory of a neutral boson field with stationary quadratic part of the Hamiltonian and nonstationary source is considered. The average number of particles produced by this source from the vacuum during a finite time evolution and the inclusive probability to record a created particle are obtained. The infrared and ultraviolet asymptotics of the average density of created particles are derived. As a particular case, quantum electrodynamics with a classical current is considered. The ultraviolet and infrared asymptotics of the average number of photons are derived. The asymptotics of the average number of photons produced by the adiabatically driven current is found.


\end{abstract}

\section{Introduction}

The quantum field theories (QFTs) with quadratic Hamiltonians are the classical subject for investigation in theoretical physics. These models represent the base for perturbation theory and, per se, describe a wide range of phenomena. It is not surprising that there is a huge literature devoted to this subject (see, e.g., the books \cite{Friedrichs,BerezMSQ1.4,BirDav.11,GrMuRaf,FullingAQFT,GFSh.3,BuchOdinShap.11,GriMaMos.11,MaslShved,DeWGAQFT.11,CalzHu,ParkTom}). Although it appears that the problem of description of QFTs with quadratic Hamiltonians having linear equations of motion is somewhat trivial, the presence of infinite degrees of freedom considerably complicates the issue. There are different methods to tackle this problem and perhaps the most straightforward Hamiltonian approach is not the commonly used one. Our aim is to fill this gap. We shall obtain the solution of this problem using the Hamiltonian formalism, i.e., we shall find the matrix elements of the finite time evolution operator generated by the nonstationary Hamiltonian of a general form for both bosons and fermions imposing rather mild assumption on the parameters of the Hamiltonian. In other words, we shall obtain the solution of the Cauchy problem for the quantum-field Schr\"{o}dinger equation. To this end, we shall modify the theory in the ultraviolet and/or infrared regions unaccessible for experiments in such a way that the finite time unitary evolution exists. Of course, when one discusses the existence of a certain QFT, not only the Hamiltonian and the algebra of observables should be specified but also their representation must be given. Different representations of the same algebra may be unitary inequivalent. We adopt in this paper the representation of the algebra of observables in the Fock space that is specified by the Hamiltonian diagonalization procedure (see, e.g., \cite{Bogol47,Friedrichs}). The physical arguments in favor of this representations will be presented below.

The fact that we consider the evolution of QFT with nonstationary Hamiltonian during a finite interval of time plays a crucial role. Nonstationarity gives rise to new physical phenomena and mathematical issues that are absent in the stationary case. The most prominent problem is that, in many cases, the background nonstationary fields complying with all the physically reasonable requirements such as an infinite smoothness, a compact support or a rapid decrease at spatial infinity, and a finiteness of the spatial volume of the system studied lead to the infinite particle production and so to the nonunitary dynamics\footnote{Throughout this paper, we call the evolution unitary if the evolution operator at any finite time is a unitary map in the separable Hilbert space of states.} \cite{Friedrichs,MaslShved,Ruijsen,NencSchr,Scharf79,Junker,Szpak15,DeckDMSch,uqem}. The presence of this problem depends on the choice of the representation of the algebra of observables and the various approaches were elaborated to overcome it \cite{CalzHu,Scharf79,Junker,uqem,ParkTom,GriMaMos.11,FullingAQFT,Friedrichs,BirDav.11,GriPav,Parkrev,KayWald,GrMa,Park1,Shirok,Shirok0,Imamur,AHHLH}. As we have already mentioned, we use the Hamiltonian diagonalization procedure as the means to specify the representation of the algebra of observables and regularize the Hamiltonian in the region of particle energies where the theory is, in fact, unknown. The possible infinite particle production in this energy domain will be compensated by the corresponding counterdiabatic terms \cite{GRKTMM19} ensuing the adiabatic evolution for these modes. As a result, the unitary evolution operator is obtained that provides a solid basis for nonperturbative calculations. In no-regularization limit, the initial ill-defined expression is recovered. Notice that we will not consider in the present paper the situation when the unitarity of a quadratic QFT is violated by the appearance of instabilities and the corresponding phase transitions (see, e.g., \cite{GrMuRaf,MigdB}). Though it is not difficult to generalize the formalism to this case.

As is known \cite{DeWGAQFT.11,Schwing.det,Schwing.10}, the vacuum-to-vacuum amplitude of the quadratic part of some QFT on a given background with zero sources defines the one-loop correction to the effective action of the complete QFT \cite{BirDav.11,GrMuRaf,FullingAQFT,GFSh.3,BuchOdinShap.11,GriMaMos.11,DeWGAQFT.11,CalzHu,ParkTom,WeinB2} that ``sums'' an infinite number of the one-loop Feynman diagrams. Fairly often, it turns out that thereby obtained effective action is nonanalytic in the coupling constant near zero and the series of the standard perturbation theory is only an asymptotic expansion of the nonperturbative expression \cite{HeisEul,Schwing.10}. We shall also obtain the well-defined nonperturbative expression for the one-loop effective action before the removal of regularization.

Apart from the construction of unitary evolution, we shall investigate a class of observables that allow for removal of regularization even in the case when the dynamics are not unitary in this limit. Namely, we shall consider the average number of created particles recorded by the detector and the probability to record a particle by the detector. We shall find the explicit expressions for these quantities and show that for any reasonable detector they are well-defined in no-regularization limit. Here the particles are defined by the Hamiltonian diagonalization procedure. This definition of the representation of the algebra of observables in the Fock space is local in time, i.e., the representation is determined by the configuration of background fields at the present instant of time. In the stationary case, this is the standard definition of particles confirmed by numerous experiments. This is the standard definition in condensed matter physics. Besides, it is clear that such a representation must make sense in the case of background fields slowly varying in time. This fact is confirmed experimentally for many systems, too. Another reason is the adiabaticity argument. It says that if the background fields vary slowly and the system starts its evolution from the vacuum state, then the present state of the system will be close to the vacuum state of the instantaneous Hamiltonian \cite{LandLifshQM.11,Kato,JoyeThes,Nenciu,AveElg,ElgHag}, i.e., to the Fock vacuum of creation-annihilation operators that diagonalize the Hamiltonian at a given instant of time. Hence, it is reasonable to suppose that such a definition of the representation of the algebra of observables is valid, at least, in the case of the background fields smoothly depending on time. Moreover, it was proved in \cite{Scharf79} that, in the inertial reference frame in Minkowski spacetime, this representation of observables leads to the unitary evolution of a quantum Dirac field in the external classical electromagnetic field. The restrictions on the electromagnetic fields imposed in \cite{Scharf79} are fulfilled for any physically realizable system. Of course, once the unitary theory is constructed, one may choose any other unitary equivalent representation of the observables.

It should be stressed that we define particles as a mere convenient means to specify the state of a system of quantum fields in the Hilbert space. Since we consider the finite time evolution, such particles are often called virtual in the literature (see, e.g., \cite{Heitl,GinzbThPhAstr}) to distinguish them from the particles in the $in$- and $out$-states at $t=\mp\infty$. However, we will not use this nomenclature as, in the real experiments lasting a finite interval of time, any particle is virtual, albeit with small virtuality, according to this terminology. Due to the energy-time uncertainty relation, the virtual particles created from the vacuum can possess very large energies at a given instant of time but most of them quickly annihilate during the evolution.

Some comments about the energy cutoff regularization are also in order. Of course, the energy cutoff regularization procedure is ambiguous as any other regularization. This ambiguity just reflects the fact that we do not know physics for very small and very large energies. Nevertheless, this regularization possesses certain merits:
\begin{enumerate}[label=\alph*)]
    \item It is the regularization, i.e., it does not change physics at energies observable in experiments;
    \item It is nonperturbative, i.e., the initial Hamiltonian of the system is modified rather than the separate terms of the series of perturbation theory;
    \item It is gauge-invariant when properly formulated (see the example in Sec. \ref{QED_w_Clas_Sour});
    \item It preserves unitarity of QFT.
\end{enumerate}
Moreover, in Sec. \ref{Inclus_Prob}, we shall construct the observables that are well-defined in no-regularization limit. They do not contain the regularization parameter and are independent of the regulator.

The paper is organized as follows. In Sec. \ref{Gener_Form}, the general formalism is developed for description of quadratic QFTs. The regularization is defined by the energy cutoff of the generator of evolution in the time dependent basis diagonalizing the Hamiltonian of the system. This, in particular, introduces the counterdiabatic terms into the initial Hamiltonian. The evolution generated by the Hamiltonian is developed not in one Fock space but in the Hilbert bundle of such spaces with the base being the manifold of background field configurations. This bundle is equipped with the trivial connection and parallel transport that allow one to bring the evolution into one Fock space. That evolution is unitary for the regularized dynamics. The explicit formulas for the matrix elements of the evolution operator are derived in this section. In particular, the well-defined expression for the one-loop effective action is deduced. Sec. \ref{Gener_Form_Bos} is devoted to bosons, whereas Sec. \ref{Gener_Form_Ferm} is for fermions. As for bosons, many formulas appearing in Sec. \ref{Gener_Form_Bos} were already given in \cite{uqem}. In the present paper, we shall generalize them to the case of theories with sources. Notice that the energy cutoff for gauge theories should be done in some gauge that uniquely separates the physical degrees of freedom as, for example, the Coulomb gauge in quantum electrodynamics (QED). The other method to introduce the energy cutoff in a gauge invariant way is to use the Hamiltonian Becchi-Rouet-Stora-Tyutin quantization (see, e.g., \cite{HeTe}), but we will not develop this approach here. The example of a gauge theory with the energy cutoff is considered in Sec. \ref{QED_w_Clas_Sour}. In Sec. \ref{Inclus_Prob}, the detection of created particles is considered. The general formulas for the average number of particles $N_D$ recorded by the detector and for the inclusive probability $w(D)$ to record a created particle are derived. It is proved that under very mild assumptions this quantities allow for removal of regularization. In this limit, $N_D$ is finite and $w(D)\in[0,1)$. In Sec. \ref{Examples}, the simple examples of the developed formalism are investigated. In Sec. \ref{Bos_w_Clas_Sour}, the quadratic boson QFT with stationary quadratic part of the Hamiltonian and nonstationary source is studied. The general formulas obtained in Secs. \ref{Gener_Form}, \ref{Inclus_Prob} are particularized for this case. In particular, we shall find the infrared and ultraviolet asymptotics of the average number of created particles. In Sec. \ref{QED_w_Clas_Sour}, QED with a classical current is considered. This is the classical example model investigated in many papers and books \cite{Shirok,BlNord37,Glaub51,SchwinS,JauRohr,BaiKatFad,GinzbThPhAstr,GavrGit90}. A special attention is paid to the influence of a finite duration of the evolution to the observed quantities such as the inclusive probabilities and the average number of created photons. We shall find the infrared and ultraviolet asymptotics of the average number of created photons and obtain the general formula for the number of photons produced by the adiabatically driven current. In Conclusion section, the main results of the paper are summarized. In Appendix \ref{Symb_Evol_Oper_App}, the general formula for the matrix elements of the evolution operator of a quadratic QFT is derived. In comparison with \cite{KazMil1}, we shall present the detailed proof of the existence of the unitary evolution and shall generalize formulas to the fermionic case. As has been already mentioned, there is an overwhelming number of papers and books devoted to both general theory and applications considered in the present paper (see some of them \cite{Nikish70,Ritus.2,RuffVerXue,FGKS11,GelTan}). Therefore, the reference list is utterly incomplete. Only the main papers and books that are immediately related to the subject matter and are known to the author are cited.

\section{General formulas}\label{Gener_Form}
\subsection{Bosons}\label{Gener_Form_Bos}

In this section, we shall construct the quantum theory of a boson field with nonstationary Hamiltonian of a general form. The representation of the algebra of observables in the Fock space will be realized with the aid of diagonalization of the Hamiltonian. In fact, we shall generalize the results of \cite{uqem} to the case of nonstationary quadratic Hamiltonians containing a linear part with respect to the field operators. In the simplest case of a scalar field interacting with nonstationary classical current in the Minkowski spacetime in the inertial reference frame, such a procedure was presented, for example, in \cite{Friedrichs,Shirok0}. The adaptation of the general notation presented below to the case of an electromagnetic field is given in Sec. \ref{QED_w_Clas_Sour} (for a scalar field see \cite{uqem,KazMil1}).

The nonstationary quadratic Hamiltonian of a neutral boson field of a general form in the Schr\"{o}din\-ger representation is written as
\begin{equation}\label{Hamilt_quadr}
\begin{gathered}
    \hat{H}(t)=\frac12 \hat{Z}^AH_{AB}(t)\hat{Z}^B+K_A(t)\hat{Z}^A,\\
                               \hat{Z}^A=\left[
                                           \begin{array}{c}
                                             \hat{\phi}^p(\spx) \\
                                             \hat{\pi}_p(\spx) \\
                                           \end{array}
                                         \right],\qquad
    [\hat{Z}^A,\hat{Z}^B]=iJ^{AB}=\left[
                        \begin{array}{cc}
                          0 & i \\
                          -i & 0 \\
                        \end{array}
                      \right]\de^p_q\de(\spx-\spy),
\end{gathered}
\end{equation}
where $p$, $q$ are the indices numerating the field components and
\begin{equation}\label{hermiticy}
    \bar{H}_{AB}=H_{BA}=H_{AB}.
\end{equation}
Henceforth the bar over the expression means complex conjugation, the condensed notation is used, and the Einstein summation convention is implied. In particular, the index $A$ includes both discrete, $p$, and continuous, $\spx$, variables and in formula \eqref{Hamilt_quadr} summation and integration over repeated indices are understood. Notice that, in the models we are interested in, the dependence on $t$ enters into \eqref{Hamilt_quadr} only through the background fields and currents $\Phi^\mu(t)$ taken at the different times $t$. The  relations \eqref{hermiticy} follow from the requirement that $\hat{H}$ is self-adjoint. Usually $H_{AB}$ defines a positive-definite quadratic form. We will not demand this property and only will assume that $H_{AB}$ is nondegenerate and
\begin{equation}
    H^{AB}:=(H^{-1})^{AB}.
\end{equation}
Then, introducing the Schr\"{o}dinger field operators explicitly depending on time,
\begin{equation}\label{deZ}
    \de \hat{Z}_t^A:=\hat{Z}^A+H^{AB}(t)K_B(t),
\end{equation}
we have
\begin{equation}\label{Hamilt_quadr1}
    \hat{H}(t)=\frac12 \de\hat{Z}^A_t H_{AB}(t)\de\hat{Z}^B_t-\frac{1}{2}K_A(t)H^{AB}(t)K_B(t).
\end{equation}
As a result, the problem of construction of QFT with the Hamiltonian \eqref{Hamilt_quadr} is reduced to similar problem considered in \cite{uqem} but with the field operators explicitly depending on time in the Schr\"{o}dinger representation.

As in \cite{uqem}, we introduce the mode functions as the solution of the spectral problem of the self-adjoint operator $-iJ_{AB}$ with respect to the quadratic form $H_{AB}(t)$:
\begin{equation}\label{eigen_prblm}
    -iJ_{AB}\ups^B_\al(t)=\omega_\al^{-1}(t) H_{AB}(t)\ups^B_\al(t),\qquad \omega_\al^{-1}(t) \bar{\ups}^A_\al(t) H_{AB}(t)\ups^B_\al(t)>0,
\end{equation}
where
\begin{equation}
    J_{AB}=\left[\begin{array}{cc}
                          0 & -1 \\
                          1 & 0 \\
                        \end{array}
                      \right]\de^p_q\de(\spx-\spy),
\end{equation}
and $\ups_\al^A$ satisfy the boundary conditions following from the problem statement. The second condition in \eqref{eigen_prblm} specifies the splitting of modes into the positive- and negative-frequency ones. For simplicity, we assume for a while that the spectrum  $\omega_\al^{-1}$ is real-valued and discrete. Besides, for any $\La>0$ there exists a finite number of eigenvalues such that $\omega_\al<\La$ and $\omega_\al\neq0$. As a rule, these properties are satisfied for systems confined to a box of a finite volume $V$ (the volume is defined with respect to some positive definite metric $\de_{ij}$). Below, in discussing the inclusive probabilities, we will relax these requirements and will consider the limit $V\rightarrow\infty$. Notice that $\omega_\al(t)$ can be negative when $H_{AB}(t)$ is not positive definite.

For example, consider a massive scalar field on the background with the metric $g_{\mu\nu}$. Its action functional without sources is
\begin{equation}
    S[\phi]=\int d^Dx\mathcal{L}=\frac12\int d^Dx\sqrt{|g|}(\partial_\mu\phi g^{\mu\nu}\partial_\nu\phi-m^2\phi^2).
\end{equation}
The canonical momentum reads as
\begin{equation}\label{Legandre_trans}
    \pi:=\sqrt{|g|}(g^{00}\dot{\phi}+g^{0i}\partial_i\phi),
\end{equation}
where $\dot{\phi}=\partial_t\phi$. The Hamiltonian density is given by
\begin{equation}\label{Hamilt_scal}
\begin{split}
    \mathcal{H}=\pi\dot{\phi}-\mathcal{L}=\,&\frac12\Big[\frac{\pi^2}{g^{00}\sqrt{|g|}} -2\frac{\pi g^{0i}\partial_i\phi}{g^{00}} -\sqrt{|g|}(\tilde{g}^{ij}\partial_i\phi\partial_j\phi-m^2\phi^2) \Big]\\
    =\,&\frac12\Big[\frac{(\pi-\sqrt{|g|}g^{0i}\partial_i\phi)^2}{g^{00}\sqrt{|g|}}-\sqrt{|g|}(g^{ij}\partial_i\phi\partial_j\phi-m^2\phi^2 )\Big],
\end{split}
\end{equation}
where $\tilde{g}^{ij}=g^{ij}-g^{0i}g^{0j}/g^{00}=(g_{ij})^{-1}$. The quadratic form $H_{AB}$ becomes
\begin{equation}
    H_{AB}=\left[
             \begin{array}{cc}
               \partial_i\sqrt{|g|}\tilde{g}^{ij}\partial_j+\sqrt{|g|}m^2 & \partial_i\frac{g^{i0}}{g^{00}} \\
               -\frac{g^{0i}}{g^{00}}\partial_i & \frac{1}{g^{00}\sqrt{|g|}} \\
             \end{array}
           \right]\de(\spx-\spy).
\end{equation}
The eigenvalue problem \eqref{eigen_prblm} takes the form
\begin{equation}\label{eigen_prblm_scal}
    \left[
       \begin{array}{c}
         \partial_i(\sqrt{|g|}\tilde{g}^{ij}\partial_ju_\al)+\sqrt{|g|}m^2u_\al+\partial_i(\frac{g^{i0}}{g^{00}}w_\al) \\
         -\frac{g^{0i}}{g^{00}}\partial_iu_\al+\frac{w_\al}{g^{00}\sqrt{|g|}} \\
       \end{array}
     \right]=i\omega_\al\left[
                          \begin{array}{c}
                            w_\al \\
                            -u_\al \\
                          \end{array}
                        \right],\qquad\ups^A_\al(t)=\left[
                                                   \begin{array}{c}
                                                     u_\al(t) \\
                                                     w_\al(t) \\
                                                   \end{array}
                                                 \right].
\end{equation}
Combining these expressions, we come to the equations
\begin{equation}\label{KG_stat}
\begin{split}
    \Big[(\hat{p}_i+\omega_\al g_i)\sqrt{|g|}g^{ij}(\hat{p}_j+\omega_\al g_j)+\sqrt{|g|}\big(\frac{\omega^2_\al}{g_{00}}-m^2\big)\Big]&u_\al=0,\\
    w_\al=-i\sqrt{|g|}\Big[g_i g^{ij}(\hat{p}_j+\omega_\al g_j)+\frac{\omega_\al}{g_{00}}\Big]&u_\al,
\end{split}
\end{equation}
where $\hat{p}_i=-i\partial_i$. In the stationary case, when the metric does not depend on time, the first equation in \eqref{KG_stat} coincides with the Klein-Gordon equation Fourier transformed with respect to time. In this case, the first equation in \eqref{KG_stat} reads as
\begin{equation}
    \sqrt{|g|}e^{i\omega t}(\nabla_\mu\nabla^\mu+m^2)e^{-i\omega t}u_\al(\spx)=0,
\end{equation}
where $\nabla_\mu$ is a covariant derivative constructed with the aid of the space-time metric $g_{\mu\nu}$. For further examples, see Sec. \ref{QED_w_Clas_Sour}.

Now we revert to a general theory. Suppose that the following orthonormality and completeness relations are satisfied:
\begin{equation}\label{orth_complt}
    \{\ups_\al,\ups_\be\}=\{\bar{\ups}_\al,\bar{\ups}_\be\}=0,\quad\{\ups_\al,\bar{\ups}_\be\}=-i\de_{\al\be},\qquad iJ^{AB}=\sum_\al\ups_\al^{[A}\bar{\ups}_\al^{B]},
\end{equation}
where $\{\ups,w\}:=J_{AB}\ups^A w^B$. The square brackets mean antisymmetrization without the factor $1/2$. The normalization of the wave functions is chosen to be
\begin{equation}
    \bar{\ups}^A_\al H_{AB}\ups_\al^B=\ups^A_\al H_{AB}\bar{\ups}_\al^B=\omega_\al.
\end{equation}
The vectors ($\ups_\al$,$\bar{\ups}_\al$) constitute a symplectic basis. Notice that
\begin{equation}
    H_{AB}=-\sum_\al\omega_\al J_{AA'} \ups_\al^{(A'}\bar{\ups}_\al^{B')}J_{B'B},\qquad H^{AB}=\sum_\al\omega_\al^{-1}\ups_\al^{(A}\bar{\ups}_\al^{B)},
\end{equation}
where the round brackets mean symmetrization without the factor $1/2$. If $H_{AB}(t)$ is positive definite, the background fields are smooth enough, and $V$ is finite, then the above properties of the spectrum and the mode functions are fulfilled for the Hamiltonians \eqref{Hamilt_quadr} describing relativistic field theories and $\omega_\al(t)>0$ (see, e.g., \cite{Shubin}). Henceforth, we just suppose that these properties hold.

Introduce the complete set of creation-annihilation operators,
\begin{equation}\label{creaannh_oper}
    \hat{a}_\al(t):=\{\bar{\ups}_\al(t),\de\hat{Z}_t\},\qquad \hat{a}^\dag_\al(t):=\{\ups_\al(t),\de\hat{Z}_t\},
\end{equation}
specifying the representation of the field operators in the Fock space $F_t$,
\begin{equation}\label{ZinAAd}
    \de\hat{Z}^A_t=-i\sum_\al[\ups_\al^A(t)\hat{a}_\al(t)-\bar{\ups}_\al^A(t)\hat{a}^\dag_\al(t)].
\end{equation}
So long as the dependence on $t$ enters into \eqref{Hamilt_quadr} only through the background fields $\Phi(t)$, the representation of the algebra of observables in the Fock space $F_t$ is determined by the values of the background fields at the present instant of time, i.e., we have the representation in $F_{\Phi(t)}$. For brevity, in what follows we denote $F_t\equiv F_{\Phi(t)}$, $a_\al(t)\equiv a_\al(\Phi(t))$, $\ups_\al(t)\equiv\ups_\al(\Phi(t))$, etc.

In particular, substituting \eqref{ZinAAd} into \eqref{Hamilt_quadr1}, we obtain
\begin{equation}\label{Hamilton_diagonalized}
    \hat{H}=\sum_\al\Big\{\frac12\omega_\al(t)[\hat{a}^\dag_\al(t)\hat{a}_\al(t)+\hat{a}_\al(t) \hat{a}^\dag_\al(t)]-\omega_\al^{-1}(t) |\ups^A_\al(t) K_A(t)|^2\Big\}.
\end{equation}
For the Hamiltonian \eqref{Hamilton_diagonalized} to be defined in $F_t$, we introduce the ultraviolet regularization by means of the projector
\begin{equation}\label{cutoff_proj}
    P^\La_{\al\be}(t)=\theta(\La-\omega_\al(t))\de_{\al\be},
\end{equation}
which acts in the one-particle Hilbert space. When the massless theories are considered at $V\rightarrow\infty$, the regularizing projection may also include the infrared cutoff characterized by the parameter $\la$. Notice that the sharp cutoff regularization can be substituted for the smooth cutoff one, i.e.,
\begin{equation}
    P^\La_{\al\be}(t)\rightarrow f(\La-\omega_\al(t))\de_{\al\be},
\end{equation}
with some smooth function $f(\La-x)$ tending sufficiently fast to zero in the energy domain excluded by the projector $P_\La$. Hereinafter, for simplicity, we suppose that
\begin{equation}
    P_\La^2=P_\La.
\end{equation}
The regularized Hamiltonian takes the form
\begin{equation}\label{Hamilt_reg}
\begin{split}
    \hat{H}_\La(\had(t),\ha(t),t)&=\sum_\al P^\La_{\al\al}(t)\Big\{\frac12\omega_\al(t)[\hat{a}^\dag_\al(t)\hat{a}_\al(t)+\hat{a}_\al(t) \hat{a}^\dag_\al(t)]-\omega_\al^{-1}(t) |\ups^A_\al(t) K_A(t)|^2\Big\}=\\
    &=\sum_\al P^\La_{\al\al}(t)\omega_\al(t)\hat{a}^\dag_\al(t)\hat{a}_\al(t)+\sum_\al P^\La_{\al\al}(t)\Big\{\frac12\omega_\al(t)-\omega_\al^{-1}(t) |\ups^A_\al(t) K_A(t)|^2\Big\},
\end{split}
\end{equation}
Notice that, in a general case, this regularization is insufficient for the dynamics generated by \eqref{Hamilt_reg} to be unitary \cite{GriPav,Parkrev,Park1,uqem,Junker,ParkTom,CalzHu,FullingAQFT,BirDav.11}. The Hamiltonian defining a unitary evolution will be given below.

Supposing that the creation-annihilation operators \eqref{creaannh_oper} specify the representation of the same field operators $\hZ$ (the generators of the algebra of observables) in the Fock spaces $F_t$,
\begin{equation}
\begin{split}
    \hat{Z}^A&=-i\sum_\al[\ups_\al^A(t)\hat{a}_\al(t)-\bar{\ups}_\al^A(t)\hat{a}^\dag_\al(t)]-H^{AB}(t)K_B(t)=\\
    &=-i\sum_\al[\ups_\al^A(t_{in})\hat{a}_\al(t_{in})-\bar{\ups}_\al^A(t_{in})\hat{a}^\dag_\al(t_{in})]-H^{AB}(t_{in})K_B(t_{in}),
\end{split}
\end{equation}
we find the relation between the creation-annihilation operators at different instants of time in the form of a linear canonical transform:
\begin{equation}\label{Bogolyub_trn}
    \left[
      \begin{array}{c}
        \hat{a}(t) \\
        \hat{a}^\dag(t) \\
      \end{array}
    \right]=
    \left[
      \begin{array}{cc}
        F(t,t_{in}) & G(t,t_{in}) \\
        \bar{G}(t,t_{in}) & \bar{F}(t,t_{in}) \\
      \end{array}
    \right]
    \left[
      \begin{array}{c}
        \hat{a}(t_{in}) \\
        \hat{a}^\dag(t_{in}) \\
      \end{array}
    \right]+
    \left[
      \begin{array}{c}
        h(t,t_{in}) \\
        \bar{h}(t,t_{in}) \\
      \end{array}
    \right],
\end{equation}
where
\begin{equation}\label{FandG}
\begin{gathered}
    F_{\al\be}=-i\{\bar{\ups}_\al(t),\ups_\be(t_{in})\},\qquad G_{\al\be}=i\{\bar{\ups}_\al(t),\bar{\ups}_\be(t_{in})\},\\
    h_\al(t,t_{in})=\{\bar{\ups}_\al(t),H^{-1}(t)K(t)-H^{-1}(t_{in})K(t_{in})\}.
\end{gathered}
\end{equation}
Henceforth, to shorten formulas, we use the matrix notation. For example,
\begin{equation}\label{matrix_nottn}
    a\ha\equiv a_\al\ha_\al,\qquad \had C\ha\equiv \had_\al C_{\al\be}\ha_\be,\qquad\text{etc.}
\end{equation}
The creation-annihilation operators related by the transform \eqref{Bogolyub_trn} can be realized in one Fock space if and only if  $G_{\al\be}$ is Hilbert-Schmidt (HS) and $h_\al$ is square-integrable (see \cite{Shale,BerezMSQ1.4,Friedrichs} and Appendix \ref{Symb_Evol_Oper_App}), i.e.,
\begin{equation}\label{unit_criter}
    \Sp G^\dag G<\infty,\qquad \bar{h} h<\infty.
\end{equation}
As a rule, these conditions are violated in relativistic QFTs (see, e.g.,  \cite{Friedrichs,MaslShved,Ruijsen,NencSchr,Scharf79,Junker,Szpak15,DeckDMSch,uqem,Avramid19}). Therefore, it is necessary to assume that $\hat{a}_\al(t)$ act in the different Fock spaces $F_t$ labeled by $t$.

Let us consider the Hilbert bundle of Fock spaces $F_\Phi$, the base being the manifold of background fields $\Phi$. The Fock spaces are constructed by means of the Hamiltonian diagonalization procedure as it was described above. In this bundle, there exists a unitary operator of parallel transport such that \cite{uqem}
\begin{equation}\label{parall_transp_map}
    \hat{W}_{\Phi(t),\Phi(t_{in})}:F_{\Phi(t_{in})}\rightarrow F_{\Phi(t)},
\end{equation}
where
\begin{equation}\label{parall_transp_def}
    \hat{a}_\al(t)=\hat{W}_{t,t_{in}}\hat{a}_\al(t_{in})\hat{W}_{t_{in},t},\qquad |vac,t\ran:=\hat{W}_{t,t_{in}}|vac,t_{in}\ran,
\end{equation}
and $|vac,t\ran$ is the vacuum vector in $F_t$. This operator obeys the equation
\begin{equation}\label{parall_transp}
    -i\partial_t\hat{W}_{t,t_{in}}=\hat{\Ga}(\had(t),\ha(t),t)\hat{W}_{t,t_{in}},
\end{equation}
where
\begin{equation}
\begin{split}
    \hat{\Ga}(\had(t),\ha(t),t):=\,&\frac12\big[2\hat{a}^\dag(t)\{\dot{\bar{\ups}}(t),\ups(t)\} \hat{a}(t)-\hat{a}(t)\{\dot{\ups}(t),\ups(t)\}\hat{a}(t)-\hat{a}^\dag(t)\{\dot{\bar{\ups}}(t),\bar{\ups}(t)\}\hat{a}^\dag(t)\big]+\\ &+i\had(t)\{\bar{\ups}(t),\partial_t(H^{-1}(t)K(t))\} -i\ha(t)\{\ups(t),\partial_t(H^{-1}(t)K(t))\},
\end{split}
\end{equation}
and $\hat{W}_{t_{in},t_{in}}=1$. The parallel transport defines the trivial self-adjoint connection
\begin{equation}\label{connection}
    \hat{\Ga}_\mu:=\hat{a}^\dag\Big\{\frac{\de\bar{\ups}}{\de\Phi^\mu},\ups\Big\} \hat{a} -\frac12\hat{a}\Big\{\frac{\de\ups}{\de\Phi^\mu},\ups\Big\}\hat{a}
    -\frac12\hat{a}^\dag \Big\{\frac{\de\bar{\ups}}{\de\Phi^\mu},\bar{\ups}\Big\}\hat{a}^\dag +i\had\Big\{\bar{\ups},\frac{\de(H^{-1}K)}{\de\Phi^\mu}\Big\} -i\ha\Big\{\ups,\frac{\de(H^{-1}K)}{\de\Phi^\mu}\Big\}
\end{equation}
in the Hilbert bundle.

The evolution operator $\hat{U}^\La_{t,t_{in}}$ generated by \eqref{Hamilt_reg} maps the Fock space $F_{t_{in}}$ into $F_{t}$. The parallel transport operator allows one to bring the evolution into one Fock space where the measurements are performed (the scalar products are calculated). The physically measured amplitudes are the matrix elements of the operator
\begin{equation}\label{Smatr_energy_def}
    \hat{S}^\La_{t,t_{in}}:=\hat{W}_{t_{in},t}\hat{U}^\La_{t,t_{in}}
\end{equation}
in the Fock space $F_{t_{in}}$. This operator satisfies the equation
\begin{equation}\label{Smatr_energy_1}
    i\partial_t\hat{S}^\La_{t,t_{in}}=\big[\hat{H}_\La(\had(t_{in}),\ha(t_{in}),t) +\hat{\Ga}(\had(t_{in}),\ha(t_{in}),t)\big]\hat{S}^\La_{t,t_{in}}
\end{equation}
with the initial condition $\hat{S}^\La_{t_{in},t_{in}}=1$.

It turns out that for many physical systems as, for example, the relativistic QFT on nonstationary gravitational background of a general form \cite{uqem}, for the quantum fields on cosmological backgrounds \cite{Parkrev,Park1,Junker,ParkTom,CalzHu,FullingAQFT,BirDav.11,AHHLH}, or for the quantum fields interacting with singular classical sources \cite{vHove52,Friedrichs,BlNord37,BerezMSQ1.4}, the operator $\hat{S}^\La_{t,t_{in}}$ is not unitary due to creation of an infinite number of particles during the evolution. Besides, unitarity of $\hat{S}^\La_{t,t_{in}}$ can be violated in the limit of $V\rightarrow\infty$. Therefore, we define the regularized evolution operator as the solution of the equation \cite{uqem}
\begin{equation}\label{Smatr_energy_2}
    i\partial_t\hat{S}'^\La_{t,t_{in}}=\big[\hat{H}_\La(\had(t_{in}),\ha(t_{in}),t) +\hat{\Ga}_\La(\had(t_{in}),\ha(t_{in}),t)\big]\hat{S}'^\La_{t,t_{in}},
\end{equation}
with the initial condition $\hat{S}'^\La_{t_{in},t_{in}}=1$, where $\hat{\Ga}_\La$ is obtained from $\hat{\Ga}$ by the replacement
\begin{equation}
    \ha_\al(t_{in})\rightarrow P^\La_{\al\be}(t)\ha_\be(t_{in}),
\end{equation}
i.e.,
\begin{multline}
    \hat{\Ga}_\La(\had(t_{in}),\ha(t_{in}),t):=\frac12\big[2\hat{a}^\dag(t_{in})P_\La\{\dot{\bar{\ups}},\ups\}P_\La \hat{a}(t_{in})-\hat{a}(t_{in})P_\La\{\dot{\ups},\ups\}P_\La\hat{a}(t_{in})-\\ -\hat{a}^\dag(t_{in})P_\La\{\dot{\bar{\ups}},\bar{\ups}\}P_\La\hat{a}^\dag(t_{in})\big]
    +i\had(t_{in})P_\La\{\bar{\ups},\partial_t(H^{-1}K)\} -i\ha(t_{in})P_\La\{\ups,\partial_t(H^{-1}K)\},
\end{multline}
It is clear that the operator $\hat{S}'^\La_{t,t_{in}}$ is unitary under the above assumptions about the spectrum and the projector $P_\La$. The passage from \eqref{Smatr_energy_1} to \eqref{Smatr_energy_2} corresponds to addition of the counterdiabatic terms \cite{GRKTMM19} to the initial Hamiltonian. This procedure gives rise to adiabatic evolution for the modes distinguished by the projector $\tilde{P}_\La(t):=1-P_\La(t)$. The dynamics of the rest modes remain unchanged. Of course, in general, the projector $\tilde{P}_\La(t)$ picks out different modes at different times. The explicit expression for the matrix elements of the operator $\hat{S}'^\La_{t,t_{in}}$ is given in \eqref{evol_symb}. The existence conditions of the theorem \ref{evol_symb_thm} are satisfied as $P_\La$ is a finite rank projector. It should also be noted that such a modification of the Hamiltonian is not related to introduction of the so-called adiabatic vacuum \cite{Parkrev,KayWald,Park1,Junker,ParkTom,CalzHu,FullingAQFT,BirDav.11}.

It is not difficult to find the explicit expression for the Hamiltonian generating the evolution \eqref{Smatr_energy_2} and the corresponding counterdiabatic terms. By definition, the evolution operator $\hat{U}'^\La_{t,t_{in}}:F_{t_{in}}\rightarrow F_t$ is
\begin{equation}\label{evolut_reg}
    \hat{U}'^\La_{t,t_{in}}:=\hat{W}_{t,t_{in}} \hat{S}'^\La_{t,t_{in}}.
\end{equation}
It satisfies the Schr\"{o}dinger equation
\begin{equation}
    i\partial_t\hat{U}'^\La_{t,t_{in}}=\hat{H}'_\La(t)\hat{U}'^\La_{t,t_{in}}
\end{equation}
with the generator
\begin{equation}\label{Hamilt_reg1}
    \hat{H}'_\La(\had(t),\ha(t),t)=\hat{H}_\La(\had(t),\ha(t),t)+\hat{\Ga}_\La(\had(t),\ha(t),t)-\hat{\Ga}(\had(t),\ha(t),t),
\end{equation}
where
\begin{equation}\label{contrdiab}
\begin{split}
    \hat{\Ga}_\La-\hat{\Ga}=\,&\hat{a}^\dag\big(P_\La\{\dot{\bar{\ups}},\ups\}P_\La-\{\dot{\bar{\ups}},\ups\}\big)\hat{a} -\frac12\hat{a}\big(P_\La\{\dot{\ups},\ups\}P_\La-\{\dot{\ups},\ups\}\big)\hat{a}-\\ &-\frac12\hat{a}^\dag\big(P_\La\{\dot{\bar{\ups}},\bar{\ups}\}P_\La-\{\dot{\bar{\ups}},\bar{\ups}\}\big)\hat{a}^\dag -i\hat{a}^\dag\tilde{P}_\La\{\bar{\ups},\partial_t(H^{-1}K)\} +i\hat{a}\tilde{P}_\La\{\ups,\partial_t(H^{-1}K)\}.
\end{split}
\end{equation}
All the operators and functions in this expression are taken at the instant of time $t$. The regularized Hamiltonian \eqref{Hamilt_reg1} can be written in terms of the field operators $\hat{Z}^A$ if one substitutes \eqref{deZ}, \eqref{creaannh_oper} into \eqref{Hamilt_reg1}. This Hamiltonian is self-adjoint, local in time, and passes into the initial Hamiltonian \eqref{Hamilt_quadr} in no-regularization limit. The terms \eqref{contrdiab} are the counterdiabatic terms. They disappear after the removal of regularization.

Now we turn to the Heisenberg representation
\begin{equation}\label{CAO_Heis}
\begin{gathered}
    \hat{a}_\al(in)=\hat{a}_\al(t_{in}),\qquad
    \hat{a}_\al(out)=\hat{U}'^\La_{t_{in},t_{out}}\hat{a}_\al(t_{out})\hat{U}'^\La_{t_{out},t_{in}} =\hat{S}'^\La_{t_{in},t_{out}}\hat{a}_\al(in) \hat{S}'^\La_{t_{out},t_{in}}.
\end{gathered}
\end{equation}
The creation-annihilation operators $(\hat{a}_\al(in),\had_\al(in))$ and $(\hat{a}_\al(out),\had_\al(out))$ act in the same Fock space $F_{t_{in}}$. Their vacuum states are
\begin{equation}
    |\overline{in}\ran:=|vac,t_{in}\ran\in F_{t_{in}},\qquad |\overline{out}\ran:=\hat{U}'^\La_{t_{in},t_{out}}|vac,t_{out}\ran\in F_{t_{in}},
\end{equation}
where $|vac,t_{out}\ran$ is defined in \eqref{parall_transp_def}. Let
\begin{equation}\label{Heis_eqs}
    \hat{Z}^A(t):=\hat{U}'^\La_{t_{in},t}\hat{Z}^A\hat{U}'^\La_{t,t_{in}},\qquad i\dot{\hat{Z}}^A(t)=[\hat{Z}^A(t),\hat{H}'_\La(t)],
\end{equation}
where $\hat{H}'_\La(t)$ is the Hamiltonian \eqref{Hamilt_reg1} written in the Heisenberg representation. As is seen, the field operators obey the regularized Heisenberg equations. In particular, if one applies this general formalism to a massive scalar field then the operator of the scalar field in the Heisenberg representation evolves in accordance with the regularized Klein-Gordon equation which, by construction, possesses a better ultraviolet behavior. After the removal of regularization, the regularized Hamiltonian turns into the initial one and so the regularized Klein-Gordon equation passes into the usual Klein-Gordon equation.

It follows from \eqref{ZinAAd}, \eqref{CAO_Heis} that
\begin{equation}\label{Zoutin}
\begin{split}
    \de\hat{Z}^A_{t_{out}}(t_{out})&=-i[\ups^A(t_{out})\hat{a}(out)-\bar{\ups}^A(t_{out})\hat{a}^\dag(out)],\\
    \de\hat{Z}^A_{t_{in}}(t_{in})&=-i[\ups^A(t_{in})\hat{a}(in)-\bar{\ups}^A(t_{in})\hat{a}^\dag(in)].
\end{split}
\end{equation}
In order to find the relation between the $in$ and $out$ creation-annihilation operators, one needs to solve the Heisenberg equations \eqref{Heis_eqs} that have the form
\begin{equation}
    i\de\dot{\hZ}_t^A(t)=iJ^A_{\La B}(t)\partial_t(H^{-1}K)^B(t)+J^{AB}H'^\La_{BC}\de \hZ_t^C(t),
\end{equation}
where
\begin{equation}
    J^A_{\La B}(t)=J^{AC}_\La(t)J_{CB},\qquad iJ^{AB}_\La(t):=\ups^{[A}(t)P_\La(t)\bar{\ups}^{B]}(t).
\end{equation}
Let us introduce the commutator Green's function
\begin{equation}\label{comm_Gren_fnc}
\begin{split}
    \tilde{G}^{AB}_\La(t,t')&:=[\hat{Z}^A(t),\hat{Z}^B(t')]=[\de\hat{Z}^A_t(t),\de\hat{Z}^B_{t'}(t')],\\
    \tilde{G}^A_{\La B}(t,t')&=\tilde{G}^{AC}_\La(t,t') J_{CB}=i\Texp\Big\{\int_{t'}^t d\tau J^{AC}H'^\La_{CB}(\tau)\Big\}.
\end{split}
\end{equation}
It is clear that, in no-regularization limit, $\tilde{G}^{AB}_{\La}(t,t')$ tends to the commutator Green's function associated with the initial Hamiltonian $H(t)$. At a finite cutoff, $\tilde{G}^{AB}_{\La}(t,t')$ does not possess the Hadamard asymptotic form at the diagonal \cite{KayWald,Junker}. The Green's function allows one to write the solution of the Heisenberg equations
\begin{equation}\label{Zout}
    \de\hat{Z}^A_{t_{out}}(t_{out})=-i\tilde{G}^A_{\La B}(t_{out},t_{in}) \de\hat{Z}^B_{t_{in}}(t_{in}) -i\int_{t_{in}}^{t_{out}} d\tau\tilde{G}^A_{\La B}(t_{out},\tau) J^B_{\La C}(\tau) \partial_\tau\big[H^{-1}K\big]^C(\tau).
\end{equation}
Employing the orthonormality and completeness relations for the mode functions, we obtain
\begin{equation}\label{canon_trans_expl}
    \left[
       \begin{array}{c}
         \hat{a}(out) \\
         \hat{a}^\dag(out) \\
       \end{array}
     \right]=\left[
               \begin{array}{cc}
                 \Phi & \Psi \\
                 \bar{\Psi} & \bar{\Phi} \\
               \end{array}
             \right]\left[
                      \begin{array}{c}
                        \hat{a}(in) \\
                        \hat{a}^\dag(in) \\
                      \end{array}
                    \right]+\left[
                              \begin{array}{c}
                                g \\
                                \bar{g} \\
                              \end{array}
                            \right],
\end{equation}
where
\begin{equation}\label{phi_psi_expl}
\begin{gathered}
    \Phi_{\al\be}(t_{out}):=-\bar{\ups}_\al^A(t_{out})\tilde{G}^\La_{AB}(t_{out},t_{in})\ups^B_\be(t_{in}),\qquad \Psi_{\al\be}(t_{out}):=\bar{\ups}_\al^A(t_{out})\tilde{G}^\La_{AB}(t_{out},t_{in})\bar{\ups}^B_\be(t_{in}),\\
    g_\al(t_{out}):=-i\int_{t_{in}}^{t_{out}} d\tau\bar{\ups}_\al^A(t_{out})\tilde{G}^\La_{AB}(t_{out},\tau) J^B_{\La C}(\tau) \partial_\tau\big[H^{-1}K\big]^C(\tau),
\end{gathered}
\end{equation}
and $\tilde{G}^\La_{AB}=J_{AC}\tilde{G}^C_{\La B}$. Comparing \eqref{canon_trans_expl} with \eqref{unit_trans_3} and \eqref{Smatr_energy_2} with \eqref{Hamil_gener}, we see that the explicit expression for the evolution operator \eqref{evol_symb} contains the operators $\Phi$ and $\Psi$ presented in \eqref{phi_psi_expl}, where one should replace $t_{out}\rightarrow t$, the functions
\begin{equation}\label{dfchi}
\begin{gathered}
    d(t)=\sum_\al P^\La_{\al\al}(t)\Big\{\frac12\omega_\al(t)-\omega_\al^{-1}(t) |\ups^A_\al(t) K_A(t)|^2 \Big\},\qquad f(t)=iP_\La\{\bar{\ups},\partial_t(H^{-1}K)\},\\
    \chi(t)=\Phi^\dag(t)g(t)-\Psi^T(t)\bar{g}(t)=-i\int_{t_{in}}^{t} d\tau\bar{\ups}^A(t_{in})\tilde{G}^\La_{AB}(t_{in},\tau) J^B_{\La C}(\tau) \partial_\tau\big[H^{-1}K\big]^C(\tau),
\end{gathered}
\end{equation}
and the operators
\begin{equation}\label{CA_energy}
    C_{\al\be}=\omega_\al P^\La_{\al\be}+\big(P^\La\{\dot{\bar{\ups}},\ups\}P^\La\big)_{\al\be},\qquad A_{\al\be}=-\big(P^\La\{\dot{\bar{\ups}},\bar{\ups}\}P^\La\big)_{\al\be}.
\end{equation}
Thereby \eqref{evol_symb} gives the explicit expression for an arbitrary matrix element of the evolution operator $\hat{S}'^\La_{t,t_{in}}$.

The vacuum-to-vacuum amplitude is given by
\begin{equation}\label{vacuum-vacuum_expl}
    \lan vac,t_{out}|\hat{U}'^\La_{t_{out},t_{in}}|vac,t_{in}\ran=\lan vac,t_{in}|\hat{S}'^\La_{t_{out},t_{in}}|vac,t_{in}\ran =\lan \overline{out}|\overline{in}\ran.
\end{equation}
The time-dependent phases of the mode functions $\ups_\al(t)$ are not specified by \eqref{eigen_prblm}, \eqref{orth_complt}. Nevertheless, it is not difficult to show \cite{KazMil1} that \eqref{vacuum-vacuum_expl} does not depend on the choice of these phases. It is convenient to fix the phases by imposing the condition
\begin{equation}
    \{\dot{\bar{\ups}}_\al,\ups_\al\}=0.
\end{equation}
Then
\begin{equation}\label{trace_C}
    \Sp C=\sum_\al P^\La_{\al\al}\big(\omega_\al +\{\dot{\bar{\ups}}_\al,\ups_\al\}\big)=\sum_\al \omega_\al P^\La_{\al\al}.
\end{equation}
Using the general formula \eqref{vacuum-vacuum_gen}, we obtain \cite{BerezMSQ1.4,DeWGAQFT.11}
\begin{equation}\label{vacuum-vacuum_expl_ferm_bos}
    \lan \overline{out}|\overline{in}\ran=\big[\det \bar{\Phi}(t_{out})\big]^{-1/2},
\end{equation}
in the absence of sources. This form for the one-loop effective action holds only before the removal of regularization and on fulfillment the additional condition imposed on the phases of the mode functions that makes the second term in the trace \eqref{trace_C} vanish.

The average field is written as
\begin{equation}
\begin{split}
    Z^A(t)=\,&\lan in|\hZ^A(t)|in\ran=-i\int_{t_{in}}^{t} d\tau\tilde{G}^A_{\La B}(t,\tau) J^B_{\La C}(\tau) \partial_\tau\big[H^{-1}K\big]^C(\tau)-H^{AB}(t)K_B(t)=\\
    =\,&-i\int_{t_{in}}^{t} d\tau \tilde{G}^{AB}_\La(t,\tau) \big[H'^\La_{BC}(\tau) -\tilde{J}^\La_{BC}(\tau)\partial_\tau\big](H^{-1}K)^C(\tau)+\\
    &+i\tilde{G}^{A}_{\La B}(t,t_{in})H^{BC}(t_{in}) K_C(t_{in}),
\end{split}
\end{equation}
where $\tilde{J}^A_{\La B}:=\de^A_B-J^A_{\La B}$ and $\tilde{J}^\La_{AB}=J_{AC}\tilde{J}^C_{\La B}$. In no-regularization limit, we have
\begin{equation}
    Z^A(t)=-i\int_{t_{in}}^{t} d\tau\tilde{G}^{AB}(t,\tau) K_B(\tau) +i\tilde{G}^{A}_{\ B}(t,t_{in})H^{BC}(t_{in}) K_C(t_{in}).
\end{equation}
The first term describes the field created by the current $K(\tau)$ at $\tau\in[t_{in},t]$. The second term describes the evolution of the field $-(H^{-1}K)(t_{in})$ that was present in the state $|vac,t_{in}\ran$.

\subsection{Fermions}\label{Gener_Form_Ferm}

Let us consider the nonstationary quadratic theory of fermionic fields of a general form and construct the corresponding QFT by means of the Hamiltonian diagonalization procedure. This procedure is analogous to the one developed in \cite{Friedrichs,Shirok,Scharf79} for the Dirac spinors interacting with classical electromagnetic fields in the inertial reference frame. In many respects, this procedure repeats the construction of the previous section. Therefore, we only outline the main steps.

Let the Hamiltonian be
\begin{equation}\label{Hamilt_quadr_ferm}
    \hat{H}(t)=\frac12 R^i_{\ j}(t)(\hat{\psi}^\dag_i\hat{\psi}^j -\hat{\psi}^j\hat{\psi}^\dag_i) +\hpsd_i\eta^i(t) +\bar{\eta}_i(t)\hps^i ,\qquad
    [\hps^i,\hpsd_j]=\de^i_j\equiv\de_{pq}\de(\spx-\spy),
\end{equation}
where $p$, $q$ are the spinor indices and $\eta(t)$, $\bar{\eta}(t)$ are Grassmann odd functions (sources). Hereinafter, the graded commutators are implied. Suppose that the operator $R(t)$ is self-adjoint with respect to the metric $\de_{ij}$, i.e.,
\begin{equation}
    R^\dag(t)=R(t),
\end{equation}
and it does not possess zero eigenvalues. The latter condition can be relaxed but we will not investigated this possibility \cite{Szpak15,GrMuRaf}.

For example, for the Dirac fields evolving on the external electromagnetic background in the inertial reference frame
\begin{equation}
    R=eA_0+m\ga^0-\ga^0\ga^i(i\partial_i-eA_i).
\end{equation}
For smooth background fields $A_\mu(x)$ with compact spatial support, this operator is self-adjoint with respect to the standard scalar product
\begin{equation}
    \lan\vf,\psi\ran:=\int d\spx \vf^\dag(t,\spx)\psi(t,\spx)
\end{equation}
on Dirac spinors. The restriction imposed on the electromagnetic fields can be relaxed (see, e.g., \cite{NencSchr,Szpak15}).

Introducing the Schr\"{o}dinger field operators,
\begin{equation}
    \de\hps_t:=\hat{\psi}+R^{-1}(t)\eta(t),\qquad \de\hpsd_t:=\hpsd+\bar{\eta}(t)R^{-1}(t),
\end{equation}
we have
\begin{equation}\label{Hamilt_ferm}
    \hat{H}(t)=\frac12\de\hat{\psi}_t^\dag R(t) \de\hat{\psi}_t -\frac12\de\hat{\psi}_t R^T(t)\de\hat{\psi}_t^\dag-\bar{\eta}(t)R^{-1}(t)\eta(t).
\end{equation}

In order to split the complete set of eigenfunctions of the operator $R(t)$ into positive- and negative-frequency ones and to construct the representation of the algebra of observables in the Fock space, we introduce the one-particle self-adjoint charge operator $q(t)$ such that
\begin{equation}
    q^2(t)=1,\qquad [R(t),q(t)]=0.
\end{equation}
Then
\begin{equation}\label{mode_split_expl}
\begin{alignedat}{2}
    R(t)u_\al(t)&=E^{(+)}_\al(t) u_\al(t),&\qquad q(t)u_\al(t)&=u_\al(t);\\
    R(t)\ups_\al(t)&=-E^{(-)}_\al(t) \ups_\al(t),&\qquad q(t)\ups_\al(t)&=-\ups_\al(t),
\end{alignedat}
\end{equation}
where the eigenfunctions $u_\al$, $\ups_\al$ are assumed to be orthonormal.

It is clear that there is an ambiguity in the definition of the operator $q(t)$ (see the discussion in \cite{Szpak15,GrMuRaf,Scharf79}). The most natural choice of this operator is such that the projector
\begin{equation}
    (1+q(t))/2
\end{equation}
projects to the states associated with the positive eigenvalues of $R(t)$. In this case, $E^{(\pm)}_\al(t)>0$. Nevertheless, one can introduce other splittings of the set of eigenfunctions of $R(t)$ \cite{GrMuRaf,Szpak15}. In particular, by analogy with the boson fields considered in the previous section, the splitting into positive- and negative-frequency modes can be defined by the requirement
\begin{equation}\label{mode_split_ferm}
    \frac{\partial E^{(\pm)}_\al}{\partial m}>0,
\end{equation}
where $m$ is the mass of (anti)particle. From physical point of view, this requirement means that the energy of states of particles and antiparticles grows with the mass of these particles. As regards the Dirac fermions, the condition \eqref{mode_split_ferm} results in the splitting of the eigenfunctions of the operator $R(t)$ into positive- and negative-frequency ones in accordance with the sign of the Dirac scalar product
\begin{equation}
    \int d\spx u_\al^\dag(t,\spx)\ga^0u_\al(t,\spx)>0,\qquad \int d\spx \ups_\al^\dag(t,\spx)\ga^0\ups_\al(t,\spx)<0.
\end{equation}
In the stationary case, this splitting complies with the standard $i\epsilon$ prescription in the sense that if $m>0$ and $m\rightarrow m-i\epsilon$, then $m^2\rightarrow m^2-i\epsilon$ and $\sgn\im E^{(\pm)}_\al=\mp1$, i.e., the poles of the propagator corresponding to particles lie below the real axis, whereas the antiparticle poles lie above the real axis.

Various splittings \eqref{mode_split_expl} of the eigenfunctions into the positive- and negative-frequency ones that differ by redefinition of a finite number of modes lead to unitary equivalent theories (see, e.g., \cite{Szpak15} and below). Therefore, all the physically reasonable splittings of the mode functions \eqref{mode_split_expl} are unitary equivalent under natural assumptions such as the smoothness of the background fields, and the boundedness in space of the system at issue. Henceforth, we will assume that a certain splitting \eqref{mode_split_expl} is chosen. The explicit form of such a splitting will be irrelevant for our study. The only assumption is that such a splitting is determined by the configuration of the background fields at the present instant of time as in the case when it is specified by the properties of the spectrum of $R(t)$.

Let us introduce the creation-annihilation operators in the Fock space $F_t$ of fermions
\begin{equation}\label{CAO_ferm}
\begin{split}
    \ha_\al(t)&:=\lan u_\al(t),\de\hps_t\ran,\quad \had_\al(t):=\lan \bar{u}_\al(t),\de\hpsd_t\ran,\\
    \hb_\al(t)&:=\lan \ups_\al(t),\de\hps_t\ran,\quad \hbd_\al(t):=\lan \bar{\ups}_\al(t),\de\hpsd_t\ran.
\end{split}
\end{equation}
Then, employing the completeness of the mode functions, we obtain the representation of the field operators (the generators of the algebra of observables) in $F_t$:
\begin{equation}\label{field_op_ferm}
    \de\hps_t=u(t)\ha(t)+\ups(t)\hbd(t),\qquad\de\hpsd_t=\bar{\ups}(t)\hb(t)+\bar{u}(t)\had(t).
\end{equation}
Substituting \eqref{field_op_ferm} into \eqref{Hamilt_ferm} and introducing the regularization, we arrive at
\begin{equation}\label{Hamilt_ferm_regul}
    \hat{H}_\La(t)=\sum_\al\big[P^\La_{\al\al}E^{(+)}_\al(t)\had_\al(t)\ha_\al(t)+  P'^\La_{\al\al}E^{(-)}_\al(t)\hbd_\al(t)\hb_\al(t) \big]+E^\La_{vac}(t),
\end{equation}
where
\begin{equation}\label{vac_energ_sym}
\begin{split}
    E^\La_{vac}:=\,&-\sum_\al\bigg\{ P^\La_{\al\al}(t)\Big[\frac12E^{(+)}_\al(t)+\frac{\lan\eta(t),u_\al(t)\ran\lan u_\al(t),\eta(t)\ran}{E^{(+)}_\al(t)}\Big]+\\
    &+P'^\La_{\al\al}(t)\Big[\frac12E^{(-)}_\al(t)-\frac{\lan\eta(t),\ups_\al(t)\ran\lan \ups_\al(t),\eta(t)\ran}{E^{(-)}_\al(t)}\Big]\bigg\},
\end{split}
\end{equation}
and
\begin{equation}
    P^\La_{\al\be}(t):=\theta(\La-E^{(+)}_\al(t))\de_{\al\be},\qquad P'^\La_{\al\be}(t):=\theta(\La-E^{(-)}_\al(t))\de_{\al\be}.
\end{equation}
If it is necessary, the infrared cutoff can also be included into $P_\La$, $P'_\La$. Notice that the symmetric ordering of operators $\hps$, $\hpsd$ is chosen in \eqref{Hamilt_ferm}. This results in the symmetric contribution of particles and antiparticles to the vacuum energy of the Hamiltonian \eqref{Hamilt_ferm_regul}. If one takes the asymmetric Hamiltonian
\begin{equation}\label{Hamilt_quadr_ferm_asym}
    \hat{H}_a(t):=\hpsd R(t)\hps+\hpsd\eta(t) +\bar{\eta}(t)\hps,
\end{equation}
then the vacuum energy in the absence of sources has the form \cite{LandLifQED,HeisEul}
\begin{equation}\label{vac_energ_asym}
    -\sum_\al P'^\La_{\al\al}(t)E^{(-)}_\al(t),
\end{equation}
i.e. it is equal to the energy of the Dirac ``sea''. For the system of Dirac fermions confined to a domain with finite volume, the difference between the finite parts of the vacuum energies \eqref{vac_energ_sym} and \eqref{vac_energ_asym} in the absence of sources takes the form of a certain surface integral \cite{KalKaz3}. Notice that, in applying this formalism to condensed mater physics, the asymmetric definition \eqref{Hamilt_quadr_ferm_asym} of the Hamiltonian is preferable as the Dirac sea of the valence electrons is actually present.

By definition, the charge operator is
\begin{equation}
    \hat{Q}_t=\had(t)\ha(t)-\hbd(t)\hb(t).
\end{equation}
Evidently, it commutes with the Hamiltonian \eqref{Hamilt_ferm_regul}. The index $t$ of the operator marks that it is realized in the space $F_t$. Recall that as long as the dependence on $t$ enters into the Hamiltonian \eqref{Hamilt_quadr_ferm} only through the background fields $\Phi(t)$, the representation of $\hat{Q}_t$ in $F_t\equiv F_{\Phi(t)}$ is determined by the configuration of background fields at the present instant of time, viz., $\hat{Q}_t\equiv \hat{Q}_{\Phi(t)}$.

Supposing that the creation-annihilation operators \eqref{CAO_ferm} provide the representation of the same generators of the algebra of observables $\hps$, $\hpsd$ in the Fock spaces $F_t$, we deduce the relations
\begin{equation}\label{Bogolyub_trn_ferm}
    \left[
      \begin{array}{c}
        \hat{c}(t) \\
        \hat{c}^\dag(t) \\
      \end{array}
    \right]=
    \left[
      \begin{array}{cc}
        F(t,t_{in}) & G(t,t_{in}) \\
        \bar{G}(t,t_{in}) & \bar{F}(t,t_{in}) \\
      \end{array}
    \right]
    \left[
      \begin{array}{c}
        \hat{c}(t_{in}) \\
        \hat{c}^\dag(t_{in}) \\
      \end{array}
    \right]+
    \left[
      \begin{array}{c}
        h(t,t_{in}) \\
        \bar{h}(t,t_{in}) \\
      \end{array}
    \right],
\end{equation}
where, for brevity, we have introduced a unified notation for the creation-annihilation operators
\begin{equation}
    \hat{c}_a(t):=(\ha_\al(t),\hb_\be(t)),\qquad \hat{c}^\dag_a(t):=(\had_\al(t),\hbd_\be(t)),
\end{equation}
and
\begin{equation}
\begin{gathered}
    F(t,t_{in})=
    \left[
      \begin{array}{cc}
        \lan u_\al(t),u_\be(t_{in})\ran & 0 \\
        0 & \lan \bar{\ups}_\al(t),\bar{\ups}_\be(t_{in})\ran \\
      \end{array}
    \right],\\
    G(t,t_{in})=
    \left[
      \begin{array}{cc}
        0 & \lan u_\al(t),\ups_\be(t_{in})\ran \\
        \lan \bar{\ups}_\al(t),\bar{u}_\be(t_{in})\ran & 0 \\
      \end{array}
    \right],\\
    h(t,t_{in})=
    \left[
      \begin{array}{c}
        \lan u_\al(t),R^{-1}(t)\eta(t)-R^{-1}(t_{in})\eta(t_{in})\ran \\
        \lan \bar{\ups}_\al(t),\bar{R}^{-1}(t)\bar{\eta}(t)-\bar{R}^{-1}(t_{in})\bar{\eta}(t_{in})\ran \\
      \end{array}
    \right].
\end{gathered}
\end{equation}
The creation-annihilation operators related by the linear canonical transform \eqref{Bogolyub_trn_ferm} can be realized in one Fock space if and only if the conditions \eqref{unit_criter} are satisfied. As for relativistic QFTs, these conditions are not fulfilled in a general position. Therefore, as in the case of boson fields, it is necessary to introduce the Hilbert bundle of Fock spaces $F_{\Phi}$ with the base being the supermanifold of background field configurations at a given instant of time.

Introduce the parallel transport operator \eqref{parall_transp_map},
\begin{equation}\label{parall_transp_def_ferm}
    \hat{c}_a(t)=\hat{W}_{t,t_{in}}\hat{c}_a(t_{in})\hat{W}_{t_{in},t},\qquad |vac,t\ran:=\hat{W}_{t,t_{in}}|vac,t_{in}\ran,
\end{equation}
and the corresponding self-adjoint trivial connection
\begin{equation}\label{parall_transp_ferm}
    -i\partial_t\hat{W}_{t,t_{in}}=\hat{\Ga}(\hat{c}^\dag(t),\hat{c}(t),t)\hat{W}_{t,t_{in}},
\end{equation}
where
\begin{equation}
\begin{gathered}
    \hat{\Ga}(\hat{c}^\dag(t),\hat{c}(t),t):=\frac12\big[2\hat{c}^\dag(t)L(t)\hat{c}(t)+\hat{c}(t)M^\dag(t)\hat{c}(t)  +\hat{c}^\dag(t)M(t)\hat{c}^\dag(t)\big] +\hat{c}^\dag(t)f'(t) +\bar{f}'(t)\hat{c}(t),\\
    L(t)=i
    \left[
      \begin{array}{cc}
        \lan\dot{u}_\al(t),u_\be(t)\ran & 0 \\
        0 & \lan\dot{\bar{\ups}}_\al(t),\bar{\ups}_\be(t)\ran \\
      \end{array}
    \right]
    ,\qquad
    M(t)=i
    \left[
      \begin{array}{cc}
        0 & \lan\dot{u}_\al(t),\ups_\be(t)\ran \\
        \lan\dot{\bar{\ups}}_\al(t),\bar{u}_\be(t)\ran & 0 \\
      \end{array}
    \right],\\
    f'(t)=i
    \left[
      \begin{array}{c}
        \lan u_\al(t),\partial_t(R^{-1}\eta)\ran \\
        \lan \bar{\ups}_\al(t),\partial_t(\bar{R}^{-1}\bar{\eta})\ran \\
      \end{array}
    \right],
\end{gathered}
\end{equation}
and
\begin{equation}\label{connection_ferm}
\begin{split}
    \hat{\Ga}_\mu:=\,&i\hat{a}^\dag \Big\lan\frac{\de u}{\de\Phi^\mu},u\Big\ran\hat{a} +i\hat{b}^\dag\Big\lan\frac{\de\bar{\ups}}{\de\Phi^\mu} ,\bar{\ups}\Big\ran\hat{b} +i\hat{a}^\dag\Big\lan\frac{\de u}{\de\Phi^\mu},\ups\Big\ran\hat{b}^\dag
    +i\hat{b}\Big\lan\frac{\de \ups}{\de\Phi^\mu},u\Big\ran\hat{a}+\\
    &+i\hat{a}^\dag\Big\lan u,\frac{\de(R^{-1}\eta)}{\de\Phi^\mu}\Big\ran +i\hat{b}^\dag\Big\lan \bar{\ups},\frac{\de(\bar{R}^{-1}\bar{\eta})}{\de\Phi^\mu}\Big\ran
    -i\Big\lan \bar{u},\frac{\de(\bar{R}^{-1}\bar{\eta})}{\de\Phi^\mu}\Big\ran\hat{a} -i\Big\lan \ups,\frac{\de(R^{-1}\eta)}{\de\Phi^\mu}\Big\ran \hat{b}.
\end{split}
\end{equation}
Then the physically measurable amplitudes are the matrix elements of the operator
\begin{equation}\label{Smatr_energy_ferm}
    \hat{S}^\La_{t,t_{in}}:=\hat{W}_{t_{in},t}\hat{U}^\La_{t,t_{in}},\qquad \hat{S}^\La_{t,t_{in}}: F_{t_{in}}\rightarrow F_{t_{in}},
\end{equation}
where $\hat{U}^\La_{t,t_{in}}: F_{t_{in}}\rightarrow F_{t}$ is the evolution operator generated by \eqref{Hamilt_ferm_regul}. The operator \eqref{Smatr_energy_ferm} satisfies the equation
\begin{equation}\label{Smatr_energy_ferm1}
    i\partial_t\hat{S}^\La_{t,t_{in}}=\big[\hat{H}_\La(\hat{c}^\dag(t_{in}),\hat{c}(t_{in}),t) +\hat{\Ga}(\hat{c}^\dag(t_{in}),\hat{c}(t_{in}),t)\big]\hat{S}^\La_{t,t_{in}}
\end{equation}
with the initial condition $\hat{S}^\La_{t_{in},t_{in}}=1$.

It was proved in \cite{Scharf79} that the operator $\hat{S}^\La_{t,t_{in}}$ is unitary for massive fermions described by the Dirac equation for a sufficiently wide class of external electromagnetic fields. However, the example of a massive scalar field studied in \cite{uqem} shows that $\hat{S}^\La_{t,t_{in}}$ is not unitary in noninertial reference frames in the flat spacetime or in the nonstationary spacetime of a general configuration. In order to secure the unitarity of evolution, we regularize the generator of the evolution operator \eqref{Smatr_energy_ferm1}. The regularized evolution operator obeys the equation
\begin{equation}\label{Smatr_energy_ferm2}
    i\partial_t\hat{S}'^\La_{t,t_{in}}=\big[\hat{H}_\La(\hat{c}^\dag(t_{in}),\hat{c}(t_{in}),t) +\hat{\Ga}_\La(\hat{c}^\dag(t_{in}),\hat{c}(t_{in}),t)\big]\hat{S}'^\La_{t,t_{in}},
\end{equation}
with the initial condition $\hat{S}'^\La_{t_{in},t_{in}}=1$, where
\begin{equation}
\begin{split}
    \hat{\Ga}_\La(\hat{c}^\dag(t_{in}),\hat{c}(t_{in}),t):=\,&\frac12\big[2\hat{c}^\dag(t_{in})\Pi_\La(t) L(t)\Pi_\La(t)\hat{c}(t_{in})+\hat{c}(t_{in})\Pi_\La(t) M^\dag(t)\Pi_\La(t)\hat{c}(t_{in})+\\  &+\hat{c}^\dag(t_{in})\Pi_\La(t)M(t)\Pi_\La(t)\hat{c}^\dag(t_{in})\big]+\\
    &+\hat{c}^\dag(t_{in})\Pi_\La(t)f'(t) +\bar{f}'(t)\Pi_\La(t)\hat{c}(t_{in}),
\end{split}
\end{equation}
and
\begin{equation}
    \Pi_\La(t):=
    \left[
      \begin{array}{cc}
        P^\La_{\al\be}(t) & 0 \\
        0 & P'^\La_{\al\be}(t) \\
      \end{array}
    \right].
\end{equation}
The replacement $\hat{\Ga}\rightarrow\hat{\Ga}_\La$ provides the adiabatic evolution for the modes of a quantum field with the energies larger than $\La$. In virtue of the fact that $\Pi_\La(t)$ is a finite rank projector, the conditions of the theorem \ref{evol_symb_thm} are satisfied and the operator $\hat{S}'^\La_{t,t_{in}}$ is unitary under the assumption that the operator $\Phi(t)$ given in \eqref{phi_psi_expl_ferm} is nondegenerate. Notice that degeneracy of the operator $\Phi(t)$ depends on the way of splitting of the eigenfunctions of $R(t)$ into positive- and negative-frequency ones.

By definition, the regularized evolution operator, $\hat{U}'^\La_{t,t_{in}}: F_{t_{in}}\rightarrow F_{t}$, has the form \eqref{evolut_reg}. It is generated by the Hamiltonian
\begin{equation}\label{Hamilt_reg_ferm1}
    \hat{H}'_\La(\hat{c}^\dag(t),\hat{c}(t),t)=\hat{H}_\La(\hat{c}^\dag(t),\hat{c}(t),t)+\hat{\Ga}_\La(\hat{c}^\dag(t),\hat{c}(t),t)-\hat{\Ga}(\hat{c}^\dag(t),\hat{c}(t),t),
\end{equation}
where
\begin{equation}\label{contrdiab_ferm}
    \hat{\Ga}_\La-\hat{\Ga}=\frac12\big[2\hat{c}^\dag(\Pi_\La L\Pi_\La-L)\hat{c}+\hat{c}(\Pi_\La M^\dag\Pi_\La -M^\dag)\hat{c}  +\hat{c}^\dag(\Pi_\La M\Pi_\La -M)\hat{c}^\dag\big] -\hat{c}^\dag\tilde{\Pi}_\La f' -\bar{f}'\tilde{\Pi}_\La\hat{c}.
\end{equation}
All the operators and functions in this expression are taken at the instant of time $t$. The projector $\tilde{\Pi}_\La(t):=1-\Pi_\La(t)$. The terms \eqref{contrdiab_ferm} are the counterdiabatic terms. Substituting \eqref{CAO_ferm} into \eqref{Hamilt_reg_ferm1}, the regularized Hamiltonian can be written in terms of the field operators. It is self-adjoint, local in time, and turns into \eqref{Hamilt_quadr_ferm} in no-regularization limit.

In the Heisenberg representation
\begin{equation}\label{CAO_Heis_ferm}
\begin{gathered}
    \hat{c}_a(in)=\hat{c}_a(t_{in}),\qquad
    \hat{c}_a(out)=\hat{U}'^\La_{t_{in},t_{out}}\hat{c}_a(t_{out})\hat{U}'^\La_{t_{out},t_{in}} =\hat{S}'^\La_{t_{in},t_{out}}\hat{c}_a(in) \hat{S}'^\La_{t_{out},t_{in}},\\
    |\overline{in}\ran:=|vac,t_{in}\ran\in F_{t_{in}},\qquad |\overline{out}\ran:=\hat{U}'^\La_{t_{in},t_{out}}|vac,t_{out}\ran\in F_{t_{in}},
\end{gathered}
\end{equation}
where $|vac,t_{out}\ran$ is defined in \eqref{parall_transp_def_ferm}. Let
\begin{equation}\label{Heis_eqs_ferm}
    \hat{\psi}^i(t):=\hat{U}'^\La_{t_{in},t}\hat{\psi}^i\hat{U}'^\La_{t,t_{in}},\qquad i\dot{\hat{\psi}}^i(t)=[\hat{\psi}^i(t),\hat{H}'_\La(t)].
\end{equation}
In particular, in the absence of sources, $\eta=\bar{\eta}=0$, the charge operator
\begin{equation}
    \hat{Q}_{t}(t)=\hat{U}'^\La_{t_{in},t}\hat{Q}_{t}\hat{U}'^\La_{t,t_{in}}=\hat{S}'^\La_{t_{in},t}\hat{Q}_{t_{in}}\hat{S}'^\La_{t,t_{in}}=\hat{Q}_{t_{in}},
\end{equation}
since $\hat{Q}_{t_{in}}$ commutes with the evolution generator \eqref{Smatr_energy_ferm2}. In other words, the average charge of the system does not depend on time provided $\eta=\bar{\eta}=0$.

It follows from \eqref{field_op_ferm}, \eqref{CAO_Heis_ferm} that
\begin{equation}\label{psioutin}
\begin{split}
    \de\hps_{t_{out}}(t_{out})&=u(t_{out})\ha(out)+\ups(t_{out})\hbd(out),\\
    \de\hps_{t_{in}}(t_{in})&=u(t_{in})\ha(in)+\ups(t_{in})\hbd(in),
\end{split}
\end{equation}
The Heisenberg equations \eqref{Heis_eqs_ferm} are written as
\begin{equation}\label{Heis_eqs_expl}
\begin{split}
    i\de\dot{\hps}_t(t)&=i\pr_\La(t)\partial_t(R^{-1}\eta)(t)+R'_\La(t)\de \hps_t(t),\\ i\de\dot{\hps}^\dag_t(t)&=i\bar{\pr}_\La(t)\partial_t(\bar{R}^{-1}\bar{\eta})(t)-\bar{R}'_\La(t)\de \hpsd_t(t),
\end{split}
\end{equation}
where
\begin{equation}
    \pr_\La(t)=|u(t)\ran P_\La(t) \lan u(t)|+|\ups(t)\ran P'_\La(t) \lan \ups(t)|,
\end{equation}
and $\bar{\pr}_\La$ is the complex conjugate operator to $\pr_\La$. Introduce the fermionic commutator Green's function
\begin{equation}
\begin{split}
    \tilde{S}^{i}_{\La j}(t,t')&:=-i[\hps^i(t),\hpsd_j(t')]=-i[\de\hps^i_t(t),\de\hpsd_{t'j}(t')],\\
    \tilde{S}^{i}_{\La j}(t,t')&=-i\Texp\Big\{-i\int_{t'}^t d\tau R'^i_{\La j}(\tau)\Big\}.
\end{split}
\end{equation}
Recall that, in the case at hand, the graded commutator is the anticommutator. In no-regularization limit, $\tilde{S}_{\La}(t,t')$ tends to the commutator Green's function associated with the initial Hamiltonian $R(t)$. The solution of the Heisenberg equations \eqref{Heis_eqs_expl} has the form
\begin{equation}\label{psiout}
    \de\hps_{t_{out}}(t_{out})=i\tilde{S}_{\La}(t_{out},t_{in}) \de\hps_{t_{in}}(t_{in}) +i\int_{t_{in}}^{t_{out}} d\tau \tilde{S}_{\La}(t_{out},\tau) \pr_\La(\tau)\partial_\tau(R^{-1}\eta)(\tau).
\end{equation}
From \eqref{psioutin} we obtain
\begin{equation}\label{canon_trans_expl_ferm}
    \left[
       \begin{array}{c}
         \hat{c}(out) \\
         \hat{c}^\dag(out) \\
       \end{array}
     \right]=\left[
               \begin{array}{cc}
                 \Phi & \Psi \\
                 \bar{\Psi} & \bar{\Phi} \\
               \end{array}
             \right]\left[
                      \begin{array}{c}
                        \hat{c}(in) \\
                        \hat{c}^\dag(in) \\
                      \end{array}
                    \right]+\left[
                              \begin{array}{c}
                                g \\
                                \bar{g} \\
                              \end{array}
                            \right],
\end{equation}
where
\begin{equation}\label{phi_psi_expl_ferm}
\begin{gathered}
    \Phi(t_{out}):=i
    \left[
      \begin{array}{cc}
        \lan u(t_{out}),\tilde{S}_\La(t_{out},t_{in}) u(t_{in})\ran & 0 \\
        0 & -\lan \bar{\ups}(t_{out}),\bar{\tilde{S}}_\La(t_{out},t_{in}) \bar{\ups}(t_{in})\ran \\
      \end{array}
    \right],\\
    \Psi(t_{out}):=i
    \left[
      \begin{array}{cc}
        0 & \lan u(t_{out}),\tilde{S}_\La(t_{out},t_{in}) \ups(t_{in})\ran \\
        -\lan \bar{\ups}(t_{out}),\bar{\tilde{S}}_\La(t_{out},t_{in}) \bar{u}(t_{in})\ran & 0 \\
      \end{array}
    \right],\\
    g_\al(t_{out}):=i
    \left[
      \begin{array}{c}
        \lan u(t_{out}),\int_{t_{in}}^{t_{out}} d\tau\tilde{S}_\La(t_{out},\tau)  \pr_\La(\tau)\partial_\tau(R^{-1}\eta)(\tau)\ran \\
        -\lan \bar{\ups}(t_{out}),\int_{t_{in}}^{t_{out}} d\tau\bar{\tilde{S}}_\La(t_{out},\tau)  \bar{\pr}_\La(\tau)\partial_\tau(\bar{R}^{-1}\bar{\eta})(\tau)\ran \\
      \end{array}
    \right].
\end{gathered}
\end{equation}
The operators $\Phi(t)$ and $\Psi(t)$ enter into the explicit expression \eqref{evol_symb} for the evolution operator. Furthermore,
\begin{equation}
\begin{gathered}
    d(t)=E_{vac}^\La(t),\qquad
    f(t)=\Pi_\La(t)f'(t),\\
    \chi=\Phi^\dag(t)g(t)+\Psi^T(t)\bar{g}(t)=i
    \left[
      \begin{array}{c}
        \lan u(t_{in}),\int_{t_{in}}^{t} d\tau\tilde{S}_\La(t_{in},\tau)  \pr_\La(\tau)\partial_\tau(R^{-1}\eta)(\tau)\ran \\
        -\lan \bar{\ups}(t_{in}),\int_{t_{in}}^{t} d\tau\bar{\tilde{S}}_\La(t_{in},\tau)  \bar{\pr}_\La(\tau)\partial_\tau(\bar{R}^{-1}\bar{\eta})(\tau)\ran \\
      \end{array}
    \right],
\end{gathered}
\end{equation}
and
\begin{equation}\label{CA_energy_ferm}
    C=
    \left[
      \begin{array}{cc}
        P^\La_{\al\be}E^{(+)}_\be & 0 \\
        0 & P'^\La_{\al\be}E^{(-)}_\be \\
      \end{array}
    \right]+
    \Pi_\La L\Pi_\La,\qquad A=\Pi_\La M\Pi_\La.
\end{equation}
These expressions substituted into \eqref{evol_symb} give the explicit form for the matrix elements of the unitary evolution operator $\hat{S}'^\La_{t,t_{in}}$.

The vacuum-to-vacuum amplitude takes the form \eqref{vacuum-vacuum_expl}. Imposing the conditions on the phases of the mode functions
\begin{equation}
    \lan\dot{u}_\al,u_\al\ran=0,\qquad\lan\dot{\ups}_\al,\ups_\al\ran=0,
\end{equation}
and keeping in mind the relations
\begin{equation}\label{d_SpC}
\begin{split}
    d\big|_{\eta=\bar{\eta}=0}&=-\frac12\Big(\sum_\al P^\La_{\al\al}E^{(+)}_\al +\sum_\al P'^\La_{\al\al}E^{(-)}_\al \Big),\\
    \Sp C&=\sum_\al P^\La_{\al\al}(E^{(+)}_\al+i\lan\dot{u}_\al,u_\al\ran) +\sum_\al P'^\La_{\al\al}(E^{(-)}_\al+i\lan\dot{\bar{\ups}}_\al,\bar{\ups}_\al\ran),
\end{split}
\end{equation}
we obtain from the general formula \eqref{vacuum-vacuum_gen} that \cite{BerezMSQ1.4,DeWGAQFT.11}
\begin{equation}\label{vacuum-vacuum_expl_ferm}
    \lan \overline{out}|\overline{in}\ran=\big[\det \bar{\Phi}(t_{out})\big]^{1/2},
\end{equation}
for vanishing sources. Just as in the case of bosons, this formula for the one-loop effective action is valid only before the removal of regularization and on supposing the additional condition on the phases of the mode functions that removes the second terms in the trace of $C$ in \eqref{d_SpC}. Furthermore, the cancelation of the terms in the exponent in \eqref{vacuum-vacuum_gen} occurs only for the symmetric ordering of operators in \eqref{Hamilt_ferm}. As regards the asymmetric ordering \eqref{Hamilt_quadr_ferm_asym}, such a cancelation does not happen and the additional factor in \eqref{vacuum-vacuum_expl_ferm} remains.

The average field is
\begin{equation}
\begin{split}
    \psi(t)&=\lan in|\hps(t)|in\ran=i\int_{t_{in}}^{t} d\tau\tilde{S}_{\La}(t,\tau) \pr_{\La}(\tau) \partial_\tau(R^{-1}\eta)(\tau)-R^{-1}(t)\eta(t)=\\
    &=\int_{t_{in}}^{t} d\tau \tilde{S}_\La(t,\tau) \big[R'_\La(\tau) -i\tilde{\pr}_\La(\tau)\partial_\tau\big](R^{-1}\eta)(\tau) -i\tilde{S}_{\La}(t,t_{in})R^{-1}(t_{in})\eta(t_{in}),
\end{split}
\end{equation}
where $\tilde{\pr}_\La:=1-\pr_\La$. In no-regularization limit, we deduce
\begin{equation}
    \psi(t)=\int_{t_{in}}^{t} d\tau \tilde{S}(t,\tau) \eta(\tau) -i\tilde{S}(t,t_{in})R^{-1}(t_{in})\eta(t_{in}).
\end{equation}
The vacuum average $\psi(t)$ of the fermionic field is a Grassmann odd function.

\section{Inclusive probabilities}\label{Inclus_Prob}

First of all, we formalize the notion of a particle detector. It follows from the postulates of quantum theory that the detector of one particle can be characterized by some self-adjoint projector (not to be confused with $P_\La$),
\begin{equation}
    P=P^\dag,
\end{equation}
in the space of one-particle states. The number of quantum states of a particle that can be recorded by the detector can be estimated from the uncertainty relation. If $V_D$ is the volume of the detector, $\Omega_p$ is the domain of particle's momenta that can be detected by the detector, and $N_s$ is the number of spin states of the particle, then the number of quantum states that can be recorded by the detector is not larger than $N_sV_D\Omega_p/(2\pi)^3$. Thus a physically realizable detector is characterized by the projector $P$ of a finite albeit very large rank. For bosons,
\begin{equation}\label{proj_det_bos}
\begin{gathered}
    P^A_{\ B}=-i\sum_{\ga=1}^K \psi_\ga^{A}\bar{\psi}_\ga^{C}J_{CB},\qquad P^\dag J=JP,\qquad P^2=P,\qquad \Sp P=K,\\  \{\psi_\ga,\psi_{\ga'}\}=\{\bar{\psi}_\ga,\bar{\psi}_{\ga'}\}=0,\quad\{\psi_\ga,\bar{\psi}_{\ga'}\}=-i\de_{\ga\ga'},\quad \ga,\ga'=\overline{1,K}.
\end{gathered}
\end{equation}
The functions $\psi_\ga$ should be linear combinations of the positive-frequency modes $\ups_\al(t_{out})$, i.e.,
\begin{equation}
    \{\ups_\al(t_{out}),\psi_\ga\}=0,\quad\forall\al.
\end{equation}
In the basis $(\ups_\al,\bar{\ups}_\al)$, we have
\begin{equation}
    P_{\al\be}=-\sum_{\ga=1}^K\{\bar{\ups}_\al(t_{out}),\psi_\ga\}\{\bar{\psi}_\ga,\ups_\be(t_{out})\},\qquad P^\dag=P,\qquad P^2=P,\qquad\Sp P=K.
\end{equation}
For fermions,
\begin{equation}\label{proj_det_ferm}
    P=\sum_{\ga=1}^K|\psi_\ga\ran\lan \psi_\ga|,\qquad \lan\psi_\ga|\psi_{\ga'}\ran=\de_{\ga\ga'},\quad\ga,\ga'=\overline{1,K},
\end{equation}
and
\begin{equation}
    \lan \ups_\al(t_{out}),\psi^{\text{part}}_\ga\ran=0,\qquad \lan \bar{u}_\al(t_{out}),\psi^{\text{antipart}}_\ga\ran=0,\quad\forall\al.
\end{equation}
In the basis of eigenfunctions of the Hamiltonian, we obtain
\begin{equation}
    P^{\text{part}}_{\al\be}=\sum_{\ga=1}^K\lan u_\al(t_{out}),\psi_\ga\ran\lan\psi_\ga,u_\be(t_{out})\ran,\qquad P^{\text{antipart}}_{\al\be}=\sum_{\ga=1}^K\lan \bar{\ups}_\al(t_{out}),\psi_\ga\ran\lan\psi_\ga,\bar{\ups}_\be(t_{out})\ran.
\end{equation}
Henceforth, we will denote these projectors as $P$. Besides, $\tilde{P}:=1-P$.

Consider the process of particle creation from the vacuum in the nonstationary background fields
\begin{equation}\label{inclus_process}
    0\rightarrow e_\ga+X,
\end{equation}
where $e_\ga$ denotes the particle in the state $\psi_\ga$, $\ga=\overline{1,K}$, and $X$ is for other particles. The average number of particles recorded by the detector at the instant of time $t_{out}$ is, by definition,
\begin{equation}\label{aver_numb}
\begin{split}
    N_D&=\lan vac,t_{in}|\hat{U}'^\La_{t_{in},t_{out}}\had(t_{out})P\ha(t_{out})\hat{U}'^\La_{t_{out},t_{in}}|vac,t_{in}\ran=\lan \overline{in}|\had(out)P\ha(out)|\overline{in}\ran=\\
    &=\lan vac,t_{in}|\hat{S}'^\La_{t_{in},t_{out}}\had(t_{in})P\ha(t_{in})\hat{S}'^\La_{t_{out},t_{in}}|vac,t_{in}\ran=\\
    &=\Sp\big[\Psi^\dag(t_{out}) P\Psi(t_{out})\big]+\bar{g}(t_{out})Pg(t_{out}).
\end{split}
\end{equation}
Hereinafter, $\ha$ denotes the annihilation operators irrespective of the sort of particles and $\had$ are the corresponding creation operators. The explicit expressions for $\Psi$ and $g$ are given in \eqref{phi_psi_expl} and \eqref{phi_psi_expl_ferm}. The average number of recorded particles, $N_D$, is finite if and only if $P\Psi$ is HS and $Pg$ is square-integrable. Before the removal of regularization, $N_D<\infty$. Moreover, the concrete examples show that the average number of particles can diverge in no-regularization limit only for infinitely large energies of particles or, in the massless case, for particle energies tending to zero. Here it is assumed that the particles are defined by means of diagonalization of the Hamiltonian of quantum fields evolving in the smooth background fields, the background fields tending sufficiently fast to zero at spatial infinity, for example, being with a compact support.

As has been already mentioned in the Introduction, we define particles as the perturbations of the corresponding vacuum state. In the expression \eqref{aver_numb}, this state is $|vac,t_{out}\ran$. To put it differently, according to this definition, the particles are just a convenient means to specify the state of a system of quantum fields in the Hilbert space. Since in \eqref{aver_numb} the finite time evolution is considered, such particles are often called virtual in the literature \cite{Heitl,GinzbThPhAstr}. We will not use this nomenclature as, for experiments lasting a finite interval of time, any particle is virtual, albeit with small virtuality, according to this terminology.

In order to find the probability of inclusive process \eqref{inclus_process}, recall that the operator
\begin{equation}
    :\exp(-\had P\ha):
\end{equation}
is the projector to the states of Fock space that do not contain (anti)particles in the states $\psi_\ga$. Then the probability of inclusive process \eqref{inclus_process} is equal to
\begin{equation}
    w_\La(D)=\lan vac,t_{in}|\hat{U}'^\La_{t_{in},t_{out}}\big[1-:\exp(-\had(t_{out}) P\ha(t_{out})):\big]\hat{U}'^\La_{t_{out},t_{in}}|vac,t_{in}\ran.
\end{equation}
It is clear that
\begin{equation}
    0\leqslant w_\La(D)<1.
\end{equation}
The probability $w_\La(D)\neq1$ as otherwise the vacuum-to-vacuum amplitude is zero (see \eqref{evol_symb}, \eqref{vacuum-vacuum}). The latter is impossible at a finite cutoff. Recall that we assume the operator $\Phi(t)$ is nondegenerate for fermions. However, it may happen that the probability of the inclusive process \eqref{inclus_process} becomes unity after the removal of regularization. This would obviously reveal the violation of unitarity in this limit. Below, we shall prove that, in no-regularization limit,
\begin{equation}
    w(D):=\lim_{\La\rightarrow\infty}w_\La(D)\in[0,1),
\end{equation}
where it is assumed that in this limit
\begin{enumerate}
  \item The operator $\Phi^\epsilon$ is bounded, where $\epsilon=\pm1$ distinguishes the statistics of the particles;
  \item The operator $P\Psi$ is HS and $Pg$ is square-integrable.
\end{enumerate}
Notice that the boundedness of $\Phi^{-\epsilon}$ follows from \eqref{canon_defn1}. If the condition 1 is satisfied, then $\Psi$ is bounded for both boson and fermions. Therefore, the operator $P\Psi$ is HS for the finite rank projectors $P$ of the form \eqref{proj_det_bos}, \eqref{proj_det_ferm}. Of course, $P\Psi$ can be HS in the case when $P$ is not a finite rank projector. As it was mentioned above, the concrete examples show that $P\Psi$ is HS and $Pg$ is square-integrable when $P$ projects to the closed domain of energies that does not include zero and infinity.

Let us write the probability in the form
\begin{equation}\label{w_la_d}
    w_\La(D)=\lan vac,t_{in}|\hat{S}'^\La_{t_{in},t_{out}}\big[1-:\exp(-\had(t_{in}) P\ha(t_{in})):\big]\hat{S}'^\La_{t_{out},t_{in}}|vac,t_{in}\ran=:1-\tilde{w}_\La(D).
\end{equation}
In order to calculate it, it is useful to pass into the Bargmann-Fock representation (see Appendix \ref{Symb_Evol_Oper_App}). Then
\begin{equation}\label{w_la_d1}
    \tilde{w}_\La(D)=\int D\bar{a}Da D\bar{a}'Da' e^{-\bar{a}a-\bar{a}'a'} \lan 0|\hat{S}'^\La_{t_{in},t_{out}}|a\ran\lan\bar{a}|:\exp(-\had(t_{in}) P\ha(t_{in})):|a'\ran\lan\bar{a}'|\hat{S}'^{\La}_{t_{out},t_{in}}|0\ran.
\end{equation}
The explicit expression for the matrix elements of the operator $\hat{S}'^\La_{t_{out},t_{in}}$ is presented in \eqref{evol_symb}. However, in our case, we may use formula \eqref{unit_trans} as the phase of the matrix element \eqref{evol_symb} does not contribute to \eqref{w_la_d}. One should also bear in mind that formula \eqref{unit_trans} gives the matrix element of the operator $\hat{S}'^{\La\dag}_{t_{out},t_{in}}$. As a result,
\begin{multline}\label{w_la_d_expl}
    \tilde{w}_\La(D)=(\det\Phi\Phi^\dag)^{-\epsilon/2}\exp\big\{\tfrac12\bar{g}\bar{X}\bar{g} +\tfrac12 gX^Tg -\bar{g}g\big\}\times\\
    \times\int D\bar{a}Da D\bar{a}'Da' \exp\big\{\tfrac{\epsilon}{2}\bar{a}'X^\dag\bar{a}' +\epsilon(g-\bar{g}X^\dag)\bar{a}' +\epsilon a(\bar{g}-Xg) +\tfrac{\epsilon}{2}aXa +\bar{a}\tilde{P}a' -\bar{a}a -\bar{a}'a' \big\},
\end{multline}
where
\begin{equation}
    X:=\bar{\Psi}\Phi^{-1}=\epsilon X^T=\epsilon (\Phi^{-1})^T \Psi^\dag.
\end{equation}
Under the above assumptions, the operator $X$ is bounded in no-regularization limit.

The functional integral \eqref{w_la_d_expl} is the Gaussian integral of the form \eqref{func_int_G} with
\begin{equation}
    B=
    \left[
      \begin{array}{cccc}
        -X & 0 & 1 & 0 \\
        0 & 0 & -\tilde{P}^T & 1 \\
        \epsilon & -\epsilon\tilde{P} & 0 & 0 \\
        0 & \epsilon & 0 & -X^\dag \\
      \end{array}
    \right],\qquad
    F=
    \left[
      \begin{array}{c}
        \epsilon(\bar{g}-Xg) \\
        0 \\
        0 \\
        g-\bar{X}\bar{g} \\
      \end{array}
    \right].
\end{equation}
Employing the formula for the inverse of a block matrix, it is not difficult to obtain
\begin{equation}
    B^{-1}=
    \left[
      \begin{array}{cccc}
        \epsilon\tP Y^{-1}X^\dag\tP^T & \epsilon\tP X^\dag (Y^{-1})^T & \epsilon(\tilde{Y}^{-1})^T & \epsilon\tP Y^{-1} \\
        \epsilon Y^{-1} X^\dag\tP^T & \epsilon Y^{-1}X^\dag & Y^{-1}X^\dag\tP^T X & \epsilon Y^{-1} \\
        \tilde{Y}^{-1} & \epsilon X\tP X^\dag(Y^{-1})^T & \epsilon\tilde{Y}^{-1}X & \epsilon X\tP Y^{-1} \\
        (Y^{-1})^T\tP^T & (Y^{-1})^T & \epsilon(Y^{-1})^T\tP^T X & \epsilon\tP^T X\tP Y^{-1} \\
      \end{array}
    \right],
\end{equation}
where
\begin{equation}
    Y:=1-\epsilon X^\dag\tP^T X\tP,\qquad \tilde{Y}:=1-\epsilon X\tP X^\dag\tP^T.
\end{equation}
Then the exponential factor in \eqref{w_la_d_expl}, \eqref{func_int_G} is written as
\begin{multline}\label{exp_factor}
    \exp\big\{\tfrac12\bar{g}\bar{X}\bar{g} +\tfrac12 gX^Tg -\bar{g}g+\tfrac{1}{2}F^T B^{-1}F \big\}=\\
    =\exp\big\{-\bar{g}Pg +\epsilon\bar{g}PX^\dag\tP^T\tilde{Y}^{-1}XP g +\tfrac{\epsilon}{2}\bar{g}P Y^{-1}X^\dag P^T\bar{g} +\tfrac{\epsilon}{2}gP^T\tilde{Y}^{-1}XP g \big\}.
\end{multline}
As for fermions, $w_\La(D)$ at $\eta=\bar{\eta}=0$, i.e., at $g=\bar{g}=0$, is only physically meaningful. Nevertheless, the expression for $w_\La(D)$ at nonzero sources is of some value since it can be used to find the probability of the inclusive process \eqref{inclus_process} in higher orders of the perturbation theory. Strictly speaking, in this case one has to suppose that the sources $(\eta, \bar{\eta})$ entering into the operators $\hat{S}'^\La_{t_{out},t_{in}}$ and $\hat{S}'^\La_{t_{in},t_{out}}$ in \eqref{w_la_d1} are different for each operator \cite{MartSchw,Keld64,GFSh.3,DeWGAQFT.11,CalzHu,GavrGit90}. This leads to obvious changes in formulas \eqref{w_la_d_expl} and \eqref{exp_factor}. We will not present this generalization here.

Let us prove that under the above assumptions the operators $Y^{-1}$ and $\tilde{Y}^{-1}$ are bounded in no-regularization limit. For fermions,
\begin{equation}
    Y^{-1}=1-X^\dag\tP^T X\tP(1+\tP X^\dag\tP^T X\tP)^{-1}.
\end{equation}
In virtue of the fact that the operator $\tP X^\dag\tP^T X\tP$ is positive definite, the operator
\begin{equation}
    (1+\tP X^\dag\tP^T X\tP)^{-1}
\end{equation}
is bounded and, consequently, so is $Y^{-1}$. Analogously one can prove that $\tilde{Y}^{-1}$ is bounded. As for bosons, we first note that it follows from \eqref{canon_defn1} and the boundedness of $\Phi$ that
\begin{equation}
    \|X^\dag X\|=\|X\|^2<1.
\end{equation}
Hence,
\begin{equation}
    \|\tP X^\dag\tP^T X\tP\|\leqslant\|X\|^2<1.
\end{equation}
Consequently, the operator
\begin{equation}
    (1-\tP X^\dag\tP^T X\tP)^{-1}
\end{equation}
is bounded. As long as
\begin{equation}
    Y^{-1}=1+X^\dag\tP^T X\tP(1-\tP X^\dag\tP^T X\tP)^{-1},\qquad \tilde{Y}^{-1}=1+X\tP(1-\tP X^\dag\tP^T X\tP)^{-1}X^\dag\tP^T,
\end{equation}
the operators $Y^{-1}$, $\tilde{Y}^{-1}$ are also bounded.

The preexponential factor stemming from \eqref{w_la_d_expl}, \eqref{func_int_G} is given by
\begin{equation}
    (\det\bar{\Phi}\Phi^T)^{-\epsilon/2}\big[\det(1-\epsilon\tP^T X\tP X^\dag)\big]^{-\epsilon/2}.
\end{equation}
The determinant can be written as
\begin{equation}\label{det1}
\begin{split}
    \det[\bar{\Phi}\Phi^T(1-\epsilon\tP^T X\tP X^\dag) ]&= \det[1+\epsilon\bar{\Phi}\Phi^T(P^TX\tP X^\dag +X P X^\dag) ]=\\
    &=\det[1+\epsilon\Phi^TP^T(\Phi^{-1})^T\Psi^\dag\tP\Psi +\epsilon\Psi^\dag P\Psi ]=: \det(1+\Omega),
\end{split}
\end{equation}
where we have use the property of the Fredholm determinant,
\begin{equation}
    \det(1+BG)=\det(1+GB),\qquad\|B\|<\infty,\quad\|G\|_1<\infty,
\end{equation}
and the relations \eqref{canon_defn1}. We cannot remove the regularization in \eqref{det1} as the second operator in the argument of the determinant in the last expression in \eqref{det1} is only HS and not trace-class. The last operator in the argument of the determinant is trace-class. Let us introduce the regularized Fredholm determinant \cite{GohbGoldKrup}
\begin{equation}\label{det2}
    \det(1+\Omega)=e^{\Sp\Omega} \sideset{}{_2}\det(1+\Omega)=e^{\epsilon\Sp(2\Psi^\dag P\Psi-\Phi^TP^T(\Phi^{-1})^T\Psi^\dag P\Psi)}
    \sideset{}{_2}\det(1+\Omega),
\end{equation}
where we have used the relations \eqref{canon_defn1} and the properties of the trace. The regularized Hilbert-Carleman determinant is uniquely defined when $\Omega$ is HS. The operator under the trace sign on the right-hand side of \eqref{det2} remains trace-class in no-regularization limit. Therefore, the right-hand side of \eqref{det2} is well-defined after the removal of regularization. Henceforward, $\det(1+\Omega)$ means the right-hand side of \eqref{det2}. Notice that due to nondegeneracy of $\bar{\Phi}$ and $\tilde{Y}^\dag$ the determinant \eqref{det1}, \eqref{det2} does not vanish.

Thus we obtain the well-defined expression for $\tilde{w}(D)\in(0,1]$ in no-regularization limit. The point $\tilde{w}(D)=0$ is excluded, because the expression in the exponent \eqref{exp_factor} is finite and the determinant \eqref{det2} is not zero or infinity. Hence,
\begin{equation}\label{prob_inclus}
\begin{split}
    w(D)=\,&1-[\det(1+\Omega)]^{-\epsilon/2}\times\\
    &\times\exp\big\{-\bar{g}Pg +\epsilon\bar{g}PX^\dag\tP^T\tilde{Y}^{-1}XP g +\tfrac{\epsilon}{2}\bar{g}P Y^{-1}X^\dag P^T\bar{g} +\tfrac{\epsilon}{2}gP^T\tilde{Y}^{-1}XP g \big\},
\end{split}
\end{equation}
and $w(D)\in[0,1)$ provided the conditions 1 and 2 above are fulfilled. In the case when the particle creation is small, viz.,
\begin{equation}
    g\sim\e,\qquad\Psi\sim\e,\qquad\Phi\Phi^\dag=\Phi^\dag \Phi=1+O(\e^2),
\end{equation}
where $\e$ is some small parameter, then, in the leading order in $\e$,
\begin{equation}
    w(D)\approx N_D-\frac12\Sp\big[\Phi^TP^T(\Phi^{-1})^T\Psi^\dag P\Psi \big]\leqslant N_D.
\end{equation}
In the particular case, $\Psi=0$, we arrive at the formula (52) of \cite{BKL5}.

\section{Examples}\label{Examples}
\subsection{Boson field with a classical source}\label{Bos_w_Clas_Sour}

To display the formalism developed above, we shall consider, as the simplest example, the quadratic theory of a neutral boson field with a classical source and a stationary quadratic part, i.e., we suppose that the Hamiltonian of the theory has the form \eqref{Hamilt_quadr} and
\begin{equation}\label{static_qdr_part}
    \partial_t\omega_\al(t)=0,\qquad\partial_t\ups_\al(t)=0.
\end{equation}
Then the regularized Hamiltonian \eqref{Hamilt_reg1} is given by
\begin{equation}\label{Hamilt_reg_exmpl}
    \hat{H}'_\La=\frac12 \de\hZ^A_t H^\La_{AB}\de\hZ^B_t -\frac12 K_A H^{AB}_\La K_B -\de\hZ^A_t \tilde{J}^\La_{AB}[H^{-1}\dot{K}]^B,
\end{equation}
where
\begin{equation}
\begin{split}
    H^\La_{AB}&=H_{AB'}J^{B'}_{\La B}=-\sum_\al\omega_\al P^\La_{\al\al} J_{AA'} \ups_\al^{(A'}\bar{\ups}_\al^{B')}J_{B'B},\\ H_\La^{AB}&=J^{A}_{\La A'} H^{A'B}=\sum_\al\omega_\al^{-1}P^\La_{\al\al}\ups_\al^{(A}\bar{\ups}_\al^{B)}.
\end{split}
\end{equation}
The last contribution in \eqref{Hamilt_reg_exmpl} is the counterdiabatic term. It disappears in no-regularization limit. The creation-annihilation operators in the Fock spaces $F_t$ are related by the Bogolyubov transform \eqref{Bogolyub_trn} with
\begin{equation}
    F_{\al\be}=\de_{\al\be},\qquad G_{\al\be}=0,\qquad h_\al(t,t_{in})=i\omega_\al^{-1}\bar{\ups}_\al^A [K_A(t)-K_A(t_{in})].
\end{equation}
This canonical transform is unitary in one Fock space if and only if  $h_\al$ is square-integrable.

The regularized commutator Green's function \eqref{comm_Gren_fnc} takes the form
\begin{equation}
\begin{split}
    \tilde{G}_\La^{AB}(t,t')&=\sum_\al\big[ \ups_\al^A \bar{\ups}_\al^B e^{-i\omega_\al P^\La_{\al\al}(t-t')} -\bar{\ups}_\al^A \ups_\al^B e^{i\omega_\al P^\La_{\al\al}(t-t')} \big]=\\
    &=\sum_\al P^\La_{\al\al}\big[ \ups_\al^A  \bar{\ups}_\al^B e^{-i\omega_\al(t-t')} -\bar{\ups}_\al^A \ups_\al^B e^{i\omega_\al(t-t')} \big]+i\tilde{J}^{AB}_\La.
\end{split}
\end{equation}
In particular,
\begin{equation}\label{Gr_fnc_reg}
    \tilde{G}_\La^{AB'}(t,t')J^\La_{B'B}=\tilde{G}^{AB'}(t,t')J^\La_{B'B},\qquad J^\La_{AA'}\tilde{G}_\La^{A'B}(t,t')=J^\La_{AA'}\tilde{G}^{A'B}(t,t'),
\end{equation}
where $\tilde{G}^{AB}(t,t')$ is the commutator Green's function in no-regularization limit. Also we shall need the symmetric Green's function
\begin{equation}
    \bar{G}^{AB}_\La(t,t')=-\frac{i}{2}\sgn(t-t')\tilde{G}_\La^{AB}(t,t'),
\end{equation}
the positive-frequency Green's function
\begin{equation}
\begin{gathered}
    G^{(+)AB}_\La(t,t'):=-i\lan in|\big[\hZ^A(t)-\lan \hZ^A(t)\ran\big] \big[\hZ^A(t')-\lan \hZ^A(t')\ran\big]|in\ran=-i\sum_\al \ups_\al^A \bar{\ups}_\al^B e^{-i\omega_\al P^\La_{\al\al}(t-t')},\\
    \lan \hZ^A(t)\ran:=\lan in| \hZ^A(t)|in\ran,
\end{gathered}
\end{equation}
the Hadamard function
\begin{equation}
    G^{(1)AB}_\La(t,t')=i(G^{(+)AB}_\La(t,t')-\bar{G}^{(+)AB}_\La(t,t'))= \sum_\al\big[ \ups_\al^A \bar{\ups}_\al^B e^{-i\omega_\al P^\La_{\al\al}(t-t')} +\bar{\ups}_\al^A \ups_\al^B e^{i\omega_\al P^\La_{\al\al}(t-t')} \big],
\end{equation}
and the Feynman propagator
\begin{equation}
\begin{split}
    G^{AB}_\La(t,t')&:=-i\lan in|T\{\big[\hZ^A(t)-\lan \hZ^A(t)\ran\big] \big[\hZ^A(t')-\lan \hZ^A(t')\ran\big]\}|in\ran=\\
    &=\bar{G}^{AB}_\La(t,t')-\frac{i}2 G^{(1)AB}_\La(t,t')=\\
    &=-i\sum_\al\big[\theta(t-t')\ups_\al^A \bar{\ups}_\al^B e^{-i\omega_\al P^\La_{\al\al}(t-t')} +\theta(t'-t)\bar{\ups}_\al^A \ups_\al^B e^{i\omega_\al P^\La_{\al\al}(t-t')} \big].
\end{split}
\end{equation}
All these Green's functions satisfy the relations of the form \eqref{Gr_fnc_reg}.

Now we find the matrix elements of the evolution operator $\hat{S}'^\La_{t,t_{in}}$, the average number of particles \eqref{aver_numb} recorded by the detector, and the probability of the inclusive process \eqref{inclus_process}. The general formulas \eqref{phi_psi_expl}, \eqref{dfchi}, and \eqref{CA_energy} imply
\begin{equation}\label{dressed_part}
\begin{gathered}
    C_{\al\be}=\omega_\al P^\La_{\al\be},\qquad A_{\al\be}=0,\qquad f_\al(t)=-\omega_\al^{-1}P^\La_{\al\al}\bar{\ups}_\al^A\dot{K}_A(t),\\
    \Phi_{\al\be}(t)=P^\La_{\al\be}e^{-i\omega_\al(t-t_{in})}+\tilde{P}^\La_{\al\be}=(\Phi^\dag)^{-1}_{\al\be}(t)=(R_{t,t_{in}})_{\al\be}, \qquad\Psi_{\al\be}(t)=0,\\
    g_\al(t)=-i\int_{t_{in}}^t d\tau e^{-i\omega_\al(t-\tau)} f_{\al}(\tau)=i \int_{t_{in}}^t d\tau e^{-i\omega_\al(t-\tau)} \omega_\al^{-1}P^\La_{\al\al}\bar{\ups}_\al^A\dot{K}_A,
\end{gathered}
\end{equation}
where $R_{t,t_{in}}$ is the operator from the theorem \ref{evol_symb_thm}. In the case at hand, the determinant appearing in \eqref{gandc} is equal to unity. The operator $\Phi$ is bounded in no-regularization limit. The expression for $d(t)$ is the same as in the general case \eqref{dfchi}. Besides,
\begin{equation}
\begin{split}
    \chi_\al(t)&=-i\int_{t_{in}}^t d\tau e^{-i\omega_\al(t_{in}-\tau)} f_{\al}(\tau)=i \int_{t_{in}}^t d\tau e^{-i\omega_\al(t_{in}-\tau)} \omega_\al^{-1}P^\La_{\al\al}\bar{\ups}_\al^A\dot{K}_A,\\
    (\Phi^\dag)^{-1}(t)\chi(t)&=g(t),
\end{split}
\end{equation}
and
\begin{equation}\label{I_dr}
\begin{split}
    I:=-\int_{t_{in}}^{t_{out}} d t \bar{f}(\Phi^\dag)^{-1}\chi&= i\int_{t_{in}}^{t_{out}}dt\int_{t_{in}}^t d\tau \sum_\al\bar{f}_\al(t) f_\al(\tau) e^{-i\omega_\al(t-\tau)}=\\
    &=\int_{t_{in}}^{t_{out}}dt\int_{t_{in}}^t d\tau [H^{-1}_\La \dot{K}]^A(t) G^{(+)}_{AB}(t,\tau) [H^{-1}_\La \dot{K}]^B(\tau),
\end{split}
\end{equation}
where $G^{(+)}_{AB}:=J_{AC}G^{(+)CD}J_{DB}$. The last expression can be written in terms of the Feynman propagator
\begin{equation}
    I=\frac12\int_{t_{in}}^{t_{out}}dt d\tau [H^{-1}_\La \dot{K}]^A(t) G_{AB}(t,\tau) [H^{-1}_\La \dot{K}]^B(\tau).
\end{equation}
Separating the real and imaginary parts, we have
\begin{equation}
\begin{split}
    \frac12&\int_{t_{in}}^{t_{out}}dt d\tau [H^{-1}_\La \dot{K}]^A(t) \bar{G}_{AB}(t,\tau) [H^{-1}_\La \dot{K}]^B(\tau)-\\ &-\frac{i}{4}\int_{t_{in}}^{t_{out}}dt d\tau [H^{-1}_\La \dot{K}]^A(t) G^{(1)}_{AB}(t,\tau) [H^{-1}_\La \dot{K}]^B(\tau).
\end{split}
\end{equation}
Whence
\begin{equation}
    I=\frac12\int_{t_{in}}^{t_{out}}dt d\tau [H^{-1}_\La \dot{K}]^A(t) \bar{G}_{AB}(t,\tau) [H^{-1}_\La \dot{K}]^B(\tau)+\frac{i}{2}\sum_\al |g_\al(t_{out})|^2.
\end{equation}
Substituting all these expressions into the general formula \eqref{evol_symb}, we obtain the matrix element of the evolution operator $\hat{S}'^\La_{t,t_{in}}$ in the Bargmann-Fock representation
\begin{equation}\label{evol_dr}
    \tilde{U}_{t_{out},t_{in}}=\exp\Big\{ \bar{a} \Phi(t_{out}) a
    +\bar{a}g(t_{out})-a\bar{\chi}(t_{out}) +iI -i\int_{t_{in}}^{t_{out}} d td(t)\Big\}.
\end{equation}
The average number of particles \eqref{aver_numb} recorded by the detector is
\begin{equation}
    N_D=\bar{g}(t_{out})Pg(t_{out}).
\end{equation}
In particular, the average number of particles created from the vacuum becomes (cf. formula (41.3) of \cite{LandLifshQM.11})
\begin{equation}\label{numb_part_dr}
    N=\sum_\al|g_\al(t_{out})|^2=2\im I.
\end{equation}
The probability \eqref{prob_inclus} of the inclusive process \eqref{inclus_process} is written as \cite{BKL5}
\begin{equation}\label{inclus_process_prob}
    w(D)=1-e^{-N_D}.
\end{equation}

The quantity
\begin{equation}\label{c_bos}
    c(t_{out})=\tilde{U}_{t_{out},t_{in}}\big|_{a=\bar{a}=0}=\exp\Big\{iI -i\int_{t_{in}}^{t_{out}} d td(t)\Big\}
\end{equation}
is the generating functional of free Green's functions. Let us show that
\begin{equation}\label{Grn_fnc_generat}
    \frac{\de^2\ln c(t_{out})}{\de K_A(t_1)\de K_B(t_2)}=-iJ^A_{\La A'}G^{A'B}(t_1,t_2),
\end{equation}
for $t_{1,2}\in(t_{in},t_{out})$ and $t_1\neq t_2$. Indeed, integrating by parts, we have
\begin{equation}\label{iI}
\begin{split}
    iI=\,&-\frac{i}{2}\int_{t_{in}}^{t_{out}}dt d\tau K_A(t)J^A_{\La A'}G^{A'B}(t,\tau)K_B(\tau) -\frac{i}{2}\int_{t_{in}}^{t_{out}} dtK_A(t) H^{AB}_\La K_B(t)+\\
    &+\frac{i}{2}\int_{t_{in}}^{t_{out}} dt[H^{-1} \dot{K}(t)]^A J^\La_{AB}[H^{-1} K(t)]^B-\\
    &-i\int_{t_{in}}^{t_{out}} dtK_A(t)G^{A}_{\ B}(t,\tau)[H^{-1}_\La K(\tau)]^B\big|^{\tau=t_{out}}_{\tau=t_{in}}-\\
    &-i[H^{-1}_\La K(t_{out})]^A G_{AB}(t_{out},t_{in})[H^{-1}_\La K(t_{in})]^B-\\
    &-\frac12\sum_\al P^\La_{\al\al}\omega_\al^{-2}\big(|\bar{\ups}^A_\al K_A(t_{out})|^2 +|\bar{\ups}^A_\al K_A(t_{in})|^2\big).
\end{split}
\end{equation}
The second term on the first line on the right-hand side of the equality is canceled out by the same term contained in $d(t)$ (see \eqref{dfchi}, \eqref{c_bos}). It is easy to see that \eqref{Grn_fnc_generat} holds.

For comparison we present here the expression for the matrix element of the evolution operator $\hat{U}_{t_{out},t_{in}}$ that is obtained without the Hamiltonian diagonalization procedure, i.e., written in terms of the creation-annihilation operators diagonalizing the Hamiltonian \eqref{Hamilt_quadr} without the source $K_A$. For brevity, we will refer to the particles associated with these creation-annihilation operators as the bare ones, whereas the particles associated with the creation-annihilation operators $\ha(t_{out})$, $\had(t_{out})$ will be called dressed. Supposing that \eqref{static_qdr_part} is fulfilled, we deduce
\begin{equation}\label{bare_part}
\begin{gathered}
    C^b_{\al\be}=\omega_\al \de_{\al\be},\qquad A^b_{\al\be}=0,\qquad f^b_\al(t)=i\bar{\ups}_\al^A K_A(t),\\
    \Phi^b_{\al\be}(t)=\de_{\al\be}e^{-i\omega_\al(t-t_{in})}=(\Phi^{b\dag})^{-1}_{\al\be}(t)=(R^b_{t,t_{in}})_{\al\be}, \qquad\Psi^b_{\al\be}(t)=0,\\
    g^b_\al(t)=-i\int_{t_{in}}^t d\tau e^{-i\omega_\al(t-\tau)} f^b_{\al}(\tau)=\int_{t_{in}}^t d\tau e^{-i\omega_\al(t-\tau)} \bar{\ups}_\al^A K_A.
\end{gathered}
\end{equation}
Also
\begin{equation}\label{bare_part1}
\begin{gathered}
    \chi_\al^b(t)=-i\int_{t_{in}}^t d\tau e^{-i\omega_\al(t_{in}-\tau)} f^b_{\al}(\tau),\qquad (\Phi^{b\dag})^{-1}(t)\chi^b(t)=g^b(t),\qquad d^b=\sum_\al \frac12\omega_\al,\\
    iI^b= -\int_{t_{in}}^{t_{out}}dt\int_{t_{in}}^t d\tau \sum_\al\bar{f}_\al(t) f_\al(\tau) e^{-i\omega_\al(t-\tau)}=-\frac{i}{2}\int_{t_{in}}^{t_{out}}dt d\tau K_A(t)G^{AB}(t,\tau)K_B(\tau).
\end{gathered}
\end{equation}
Therefore,
\begin{equation}\label{evol_bare}
    \tilde{U}^b_{t_{out},t_{in}}=\exp\Big\{ \bar{a} \Phi^b(t_{out}) a
    +\bar{a}g^b(t_{out})-a\bar{\chi}^b(t_{out}) +iI^b -i \int^{t_{out}}_{t_{in}} dtd^b\Big\}.
\end{equation}
We shall consider the connection between bare and dressed particles in more detail below in discussing QED with a classical current. Here we only note that if $\bar{\ups}_\al^AK_A(t)$ tends to zero at $|t|\rightarrow\infty$ for all $\al$ then, after the removal of regularization,
\begin{equation}
    g_\al = g_\al^b,
\end{equation}
for $t_{in}\rightarrow-\infty$, $t_{out}\rightarrow\infty$. In this case, the average number of created particles, the average number of particles, recorded by the detector, and the probability of the inclusive process  \eqref{inclus_process} are the same for both the bare and dressed particles.

When $t_{in}$, $t_{out}$ are finite, this is not the case. If the source $K_A(t)$ is a sufficiently smooth function of $t$, i.e., for large $\omega_\al$, we have
\begin{equation}\label{current_tderiv}
    |\bar{\ups}_\al^A K_A|_{UV}\sim |\bar{\ups}_\al^A \dot{K}_A|_{UV},
\end{equation}
then the average number of bare particles created with high energies
\begin{equation}
    \sum_{\al\in UV}|g^b_\al(t_{out})|^2
\end{equation}
behaves worse than \eqref{numb_part_dr} in no-regularization limit. From \eqref{dressed_part}, \eqref{bare_part} we see that in a general position (cf. the asymptotics of (88) and (92) in \cite{uqem})
\begin{equation}\label{bare_dress}
    |g_\al(t_{out})|^2_{UV}\sim\omega_\al^{-2} |g^b_\al(t_{out})|^2_{UV}.
\end{equation}
The representation of the algebra of observables in the Hilbert bundle of Fock spaces defined by means of the Hamiltonian diagonalization procedure improves the ultraviolet behavior of the average number of particles \cite{Shirok} by the two powers of energy.

In the infrared limit, for massless particles and finite $t_{in}$, $t_{out}$, the situation is opposite. As long as the relation \eqref{bare_dress} holds, the average number of bare particles created at small energies behaves better than \eqref{numb_part_dr} in no-regularization limit. It is not hard to find the infrared asymptotics of the expressions entering into \eqref{evol_dr}, \eqref{evol_bare}. Taking into account that
\begin{equation}
    \sum_\al\sim\int d\spp,\qquad\ups_\al\sim\omega_\al^{-1/2}=|\spp|^{-1/2},
\end{equation}
it follows from \eqref{I_dr} or \eqref{iI} for dressed particles in no-regularization limit
\begin{equation}\label{IR_dr}
    \Big(iI-i\int_{t_{in}}^{t_{out}} d td(t)\Big)_{IR}=-\frac12\sum_{\al\in IR} \Big\{\Big|\omega_\al^{-1}\bar{\ups}_\al^AK_A\big|_{t_{in}}^{t_{out}}\Big|^2 +\int_{t_{in}}^{t_{out}} \frac{dt}{\omega_\al^2} \dot{K}_A(t) \ups_\al^{[A}\bar{\ups}_\al^{B]} K_B(t) \Big\},
\end{equation}
where the summation is carried out over the quantum numbers $\al$ with the energies $\omega_\al$ much less than all other energy scales and it is assumed that $\omega_\al(t_{out}-t_{in})\ll1$. The last condition implies that the radiation is not formed at the energies $\omega_\al$. The second term in \eqref{IR_dr} is negligibly small in comparison with the first one as the mode functions corresponding to zero energy can always be chosen real-valued (see \eqref{eigen_prblm}), and so
\begin{equation}
    \ups_\al^{[A}\bar{\ups}_\al^{B]}=o(\omega_\al^{-1}),
\end{equation}
for $\omega_\al\rightarrow0$. As a result,
\begin{equation}
    \Big(iI-i\int_{t_{in}}^{t_{out}} d td(t)\Big)_{IR}=-\frac12\sum_{\al\in IR} \Big|\omega_\al^{-1}\bar{\ups}_\al^AK_A\big|_{t_{in}}^{t_{out}}\Big|^2.
\end{equation}
The infrared asymptotics of created dressed particles is written as
\begin{equation}
    N_{IR}=\sum_{\al\in IR} |g_\al(t_{out})|^2= \sum_{\al\in IR} \Big|\omega_\al^{-1}\bar{\ups}_\al^AK_A\big|_{t_{in}}^{t_{out}}\Big|^2= \sum_{\al\in IR}|h_\al(t_{out},t_{in})|^2.
\end{equation}
Therefore, the dynamics in the infrared region are unitary if and only if the canonical transforms \eqref{Bogolyub_trn} define the unitary transforms in one Fock space in the infrared region. In other words, the use of Hilbert bundle of Fock spaces defined by the Hamiltonian diagonalization does not improve the infrared behavior of dynamics of a massless field. From physical point of view, this fact is not a trouble as one can always suppose that the system at issue is confined into a large box.

As for bare particles, we have
\begin{equation}
    iI^b_{IR}=-\frac12\sum_{\al\in IR}\Big|\int_{t_{in}}^{t_{out}}dt\bar{\ups}_\al^AK_A(t) \Big|^2,
\end{equation}
and
\begin{equation}
    N^b_{IR}=\sum_{\al\in IR} |g^b_\al(t_{out})|^2= \sum_{\al\in IR}\Big|\int_{t_{in}}^{t_{out}}dt\bar{\ups}_\al^AK_A(t) \Big|^2.
\end{equation}
For the space dimension $d\geqslant2$, the number of particles $N_{IR}$ is finite for smooth sources $K_A(t)$ tending sufficiently fast to zero at spatial infinity.

\subsection{Quantum electrodynamics with a classical current}\label{QED_w_Clas_Sour}

Let us apply the above general formulas to QED with a classical current in the Minkowski spacetime in the inertial reference frame \cite{Shirok,BlNord37,Glaub51,SchwinS,JauRohr,BaiKatFad,GinzbThPhAstr,GavrGit90}. The Minkowski metric is
\begin{equation}
    \eta_{\mu\nu}=diag(-1,1,1,1).
\end{equation}
The Hamiltonian of the electromagnetic field in the Coulomb gauge reads as (see, e.g., \cite{WeinbergB.12,GinzbThPhAstr,Heitl})
\begin{equation}\label{QED_Hamilt}
    \hat{H}=\int d\spx\big[\frac{1}{2}\hat{\pi}^2_i +\frac12 \hat{A}_i\rot^2_{ij}\hat{A}_j +\hat{A}_i j^i_\perp\big]+V_{\text{Coul}},\qquad V_{\text{Coul}}=-\frac{1}{2}j^0\Delta^{-1}j^0,
\end{equation}
where $j^\mu(x)$ is the conserved classical current,
\begin{equation}
    \partial_\mu j^\mu(x)=0,
\end{equation}
and
\begin{equation}
    j^i_\perp=j^i-\partial_i\De^{-1}\partial_j j^j=:\de_{\perp j}^ij^j.
\end{equation}
In the Coulomb gauge,
\begin{equation}
    \partial_i\hat{\pi}_i=0,\qquad \partial_i\hat{A}_i=0.
\end{equation}
The canonical commutation relations are
\begin{equation}
    [\hat{A}_i(\spx),\hat{\pi}_j(\spy)]=i\de^\perp_{ij}(\spx-\spy).
\end{equation}
Using the notation from \eqref{Hamilt_quadr}, we have
\begin{equation}
\begin{gathered}
    \hat{Z}^A=
    \left[
      \begin{array}{c}
        \hat{A}_i(\spx) \\
        \hat{\pi}_i(\spx) \\
      \end{array}
    \right],\qquad
    K_A=
    \left[
      \begin{array}{c}
        j^i(x) \\
        0 \\
      \end{array}
    \right],\qquad
    H_{AB}=
    \left[
      \begin{array}{cc}
        \rot^2_{ij} & 0 \\
        0 & \de^\perp_{ij} \\
      \end{array}
    \right],\\
    J_{AB}=
    \left[
      \begin{array}{cc}
        0 & -1 \\
        1 & 0 \\
      \end{array}
    \right]\de^\perp_{ij},\qquad
    J^{AB}=
    \left[
      \begin{array}{cc}
        0 & 1 \\
        -1 & 0 \\
      \end{array}
    \right]\de^\perp_{ij}.
\end{gathered}
\end{equation}
Introducing the notation for the components of the mode functions,
\begin{equation}
    \ups^A_\al=
    \left[
      \begin{array}{c}
        u^i_\al(\spx) \\
        w^i_\al(\spx) \\
      \end{array}
    \right],
\end{equation}
the complete set of solutions \eqref{eigen_prblm}, \eqref{orth_complt} can be taken in the form of plane waves
\begin{equation}
    \mathbf{u}_\al=\frac{\spe_{(s)}(\spk)}{\sqrt{2|\spk|V}}e^{i\spk\spx},\qquad \mathbf{w}_\al=-i \sqrt{\frac{|\spk|}{2V}} \spe_{(s)}(\spk)e^{i\spk\spx},\qquad\omega_\al=|\spk|,\qquad\sum_\al\equiv\sum_s\int\frac{Vd\spk}{(2\pi)^3},
\end{equation}
where $V$ is the normalization volume,  $\al=(s,\spk)$, $s=\overline{1,2}$, and
\begin{equation}
    \spe_{(s)}(\spk)\spk=0,\qquad\sum_s e^{(s)}_i(\spk)\bar{e}^{(s)}_j(\spk)=\de_{ij}-k_ik_j/\spk^2=\de_{ij}^\perp.
\end{equation}
Then
\begin{equation}
    H^{AB}K_B(t)=\int\frac{d\spk}{(2\pi)^3\spk^2}
    \left[
      \begin{array}{c}
        e^{-i\spk\spx}j^i_\perp(t,\spk) \\
        0 \\
      \end{array}
    \right],
\end{equation}
where
\begin{equation}
    j^i(t,\spk):=\int d\spx e^{i\spk\spx} j^i(t,\spx),\qquad \bar{j}^i(t,\spk)=j^i(t,-\spk).
\end{equation}
In particular, for charged point particles
\begin{equation}\label{curr_part}
    j^i(t,\spk)=\sum_n e_n\beta^i_n(t) e^{i\spk\spx_n(t)},\qquad j^0(t,\spk)=\sum_n e_n e^{i\spk\spx_n(t)},
\end{equation}
where $e_n$ is the charge of the $n$-th particle and $\beta^i_n$ is its velocity. The Schr\"{o}dinger field operators \eqref{deZ} are written as
\begin{equation}\label{deZ_qed}
    \hat{A}_i(\spx)=\de\hat{A}_{ti}(\spx)-\int\frac{d\spk}{(2\pi)^3\spk^2}e^{-i\spk\spx}j^\perp_{i}(t,\spk).
\end{equation}
The last term is nothing but the Biot-Savart field \cite{LandLifshCTF} produced by the current $j^\perp_{i}$. At large distances from the source, in the wave zone, this contribution tends to zero and the operators $\hat{A}_i(\spx)$ and $\de\hat{A}_{ti}(\spx)$ coincide.

The creation-annihilation operators $\ha_\al(t)$, $\had_\al(t)$ at different times are related by the canonical transform \eqref{Bogolyub_trn} with
\begin{equation}
    F_{\al\be}=\de_{\al\be},\qquad G_{\al\be}=0,\qquad h_\al(t,t_{in})=\frac{i}{\omega_\al}\frac{\bar{e}_i^{(s)}(\spk)[\bar{j}^i(t,\spk)-\bar{j}^i(t_{in},\spk)]}{\sqrt{2\omega_\al V}}.
\end{equation}
Therefore,
\begin{equation}
    \sum_\al|h_\al(t,t_{in})|^2=\int\frac{d\spk}{(2\pi)^3}\frac{\de^\perp_{ij}}{2|\spk|^3}[\bar{j}^i(t,\spk)-\bar{j}^i(t_{in},\spk)] [j^j(t,\spk)-j^j(t_{in},\spk)].
\end{equation}
In particular, the relation between the annihilation operators of bare and dressed particles is
\begin{equation}\label{bare_dress_qed}
    \ha_\al(t)=\hat{b}_\al+\frac{i}{\omega_\al}\frac{\bar{e}_i^{(s)}(\spk)\bar{j}^i(t,\spk)}{\sqrt{2\omega_\al V}},
\end{equation}
where $\hat{b}_\al$ are the annihilation operators of bare photons. Taking into account that
\begin{equation}\label{em_field}
\begin{split}
    \de\hat{A}_{ti}(\spx)&=-i\sum_s\int\frac{d\spk}{(2\pi)^3}\sqrt{\frac{V}{2|\spk|}} \big[e^{(s)}_i(\spk)e^{i\spk\spx}\ha_{(s)}(\spk;t) -\bar{e}^{(s)}_i(\spk)e^{-i\spk\spx}\had_{(s)}(\spk;t)\big],\\
    \hat{A}_{i}(\spx)&=-i\sum_s\int\frac{d\spk}{(2\pi)^3}\sqrt{\frac{V}{2|\spk|}} \big[e^{(s)}_i(\spk)e^{i\spk\spx}\hat{b}_{(s)}(\spk) -\bar{e}^{(s)}_i(\spk)e^{-i\spk\spx}\hat{b}^\dag_{(s)}(\spk)\big],
\end{split}
\end{equation}
it follows from \eqref{deZ_qed} that the excitations of the quantum electromagnetic field described by bare and dressed photons almost coincide in the wave zone. The canonical transform \eqref{bare_dress_qed} is unitary provided
\begin{equation}
    \int\frac{d\spk}{(2\pi)^3} \frac{\bar{j}_i^\perp(t,\spk) j_i^\perp(t,\spk)}{2|\spk|^3}<\infty,
\end{equation}
i.e., when the average number of bare photons in the Biot-Savart field is finite (see \eqref{deZ_qed}). Notice that in the case of a stationary current, $\dot{j}^\perp_i(t)=0$, it is the states of the Fock basis constructed by the use of the operators $\hat{a}^\dag_\al(t)$ acting on their vacuum, which are stationary. To put it differently, in this case the stable particles are the dressed photons rather than the bare ones. In the stationary case at a finite temperature, the dressed photons, and not the bare ones, are distributed over the energies in accordance with the Bose-Einstein distribution. On the other hand, the bare photons enter into the decomposition of the quantum electromagnetic field \eqref{em_field} and, in this sense, it is these particles which interact with other fields in the theory that are not included into the Hamiltonian \eqref{QED_Hamilt}. The shift \eqref{bare_dress_qed} results in that the other fields of the theory interact with the classical Biot-Savart field plus the quantum perturbations described by the dressed photons.

Introducing the regularization as in \eqref{Hamilt_reg_exmpl} and using the general formulas \eqref{dressed_part}, we come to
\begin{equation}\label{dressed_part_qed}
\begin{gathered}
    f_\al(t)=-P^\La_{\al\al}\frac{\bar{e}_i^{(s)}\dot{\bar{j}}^i(t,\spk)}{\sqrt{2V}|\spk|^{3/2}},\qquad
    g_\al(t)=i\int_{t_{in}}^t d\tau e^{-i|\spk|(t-\tau)} P^\La_{\al\al}\frac{\bar{e}_i^{(s)}\dot{\bar{j}}^i(t,\spk)}{\sqrt{2 V}|\spk|^{3/2}},\\
    d(t)=\int\frac{d\spk}{(2\pi)^3}P^\La_{\al\al}\Big[ V|\spk| +\frac{|j_0(t,\spk)|^2-|j^\perp_i(t,\spk)|^2}{2\spk^2}\Big].
\end{gathered}
\end{equation}
The expressions for the operators $C$, $A$, $\Phi$, and $\Psi$ are the same as in \eqref{dressed_part}. Besides,
\begin{equation}
    \chi_\al(t)=i\int_{t_{in}}^t d\tau e^{-i|\spk|(t_{in}-\tau)} P^\La_{\al\al}\frac{\bar{e}_i^{(s)}\dot{\bar{j}}^i(t,\spk)}{\sqrt{2 V}|\spk|^{3/2}},
\end{equation}
and
\begin{equation}\label{I_qed}
\begin{split}
    I=\,&i\int_{t_{in}}^{t_{out}}dt\int_{t_{in}}^t d\tau \int\frac{d\spk}{(2\pi)^3}P^\La_{\al\al} \frac{\dot{j}^\perp_i(t,\spk) \dot{\bar{j}}^\perp_i(\tau,\spk)}{2|\spk|^3} e^{-i|\spk|(t-\tau)}=
    \frac{i}{4} \int_{t_{in}}^{t_{out}}dt d\tau\sgn(t-\tau)\times\\
    &\times\int\frac{d\spk}{(2\pi)^3} \frac{P^\La_{\al\al}}{2|\spk|^3} \big[\dot{j}^\perp_i(t,\spk) \dot{\bar{j}}^\perp_i(\tau,\spk)e^{-i|\spk|(t-\tau)} -\dot{\bar{j}}^\perp_i(t,\spk) \dot{j}^\perp_i(\tau,\spk)e^{i|\spk|(t-\tau)}\big]+\\
    &+\frac{i}{2}\int\frac{d\spk}{(2\pi)^3}\frac{P^\La_{\al\al}}{2|\spk|^3} \Big|\int_{t_{in}}^{t_{out}}dt e^{i|\spk|t}\dot{j}^\perp_i(t,\spk) \Big|^2.
\end{split}
\end{equation}
Substituting these expressions into \eqref{evol_dr}, we obtain the matrix element of the evolution operator $\hat{S}'^\La_{t,t_{in}}$. The average number of photons recorded by the detector and the probability of the inclusive process \eqref{inclus_process} are given by the formulas \eqref{numb_part_dr}, \eqref{inclus_process_prob}.

Let us provide the physical interpretation to the derived formulas. The first term in $d(t)$ is the energy of vacuum fluctuations. The second term in $d(t)$ is the energy of a Coulomb interaction. The third term in $d(t)$ is the energy of interaction due to the Biot-Savart field. This quantity is negative (see, e.g., \cite{LandLifshECM}) as it includes not only the energy of the magnetic field but also the energy of interaction of this field with the current. The average number of created dressed photons is (\cite{Shirok}, see also formula (41.3) of \cite{LandLifshQM.11})
\begin{equation}\label{aver_numb_dr_qed}
    N=\sum_\al|g_\al(t_{out})|^2=\int\frac{d\spk}{(2\pi)^3}\frac{P^\La_{\al\al}}{2|\spk|^3} \Big|\int_{t_{in}}^{t_{out}}dt e^{-i|\spk|t}\dot{j}^\perp_i(t,\spk) \Big|^2.
\end{equation}
As it was noted in the previous section, for $t_{in}\rightarrow-\infty$, $t_{out}\rightarrow\infty$, the integration by parts turns this formula into the standard formula for the average number of photons radiated by a classical current \cite{LandLifshCTF}. In particular, $|g_\al|^2$ possesses the standard infrared asymptotics \cite{WeinbIR,LandLifshCTF,WeinbergB.12,GinzbThPhAstr,JauRohr,BlNord37} provided the trajectories of charged particles in the $in$ and $out$ regions tend to a uniform rectilinear motion. For this asymptotics takes place, it is assumed that $|\spk|(1-\beta_n)(t_{out}-t_{in})\gg1$, i.e., the radiation has time to form at a given energy. The change of phase of the wave function of the system (the Coulomb phase) during the infinite interval of time, $(t_{out}-t_{in})$, also becomes divergent in the infrared domain.

For finite $t_{in}$, $t_{out}$ the quantity \eqref{aver_numb_dr_qed} determines the average number of dressed photons in the quantum state of the field at the instant of time $t_{out}$ in the following experimental setup. For $t\leqslant t_{in}$ the stationary system, $\dot{j}^\perp_i(t)=0$, is in the ground (vacuum) state. Then for $t\in(t_{in},t_{out})$ the classical current, $j^\perp_i(x)$, is changing. At the instant of time $t=t_{out}$, the detector records the number of dressed photons and is turned off or for $t\geqslant t_{out}$ the current does not depend on time. Of course, in order to measure the average number of photons, one needs to carry out a series of identical experiments. The real detector cannot precisely measure the quantity \eqref{aver_numb_dr_qed} or its density for any momentum as it was discussed in Sec. \ref{Inclus_Prob}. If $\tau_s$ is the typical switching off time of the detector, then the detector can measure the density of \eqref{aver_numb_dr_qed} for the photon energies $|\spk|\tau_s\ll1$. The very quantity \eqref{aver_numb_dr_qed} is independent of the detector characteristics and it is the question of the detector design for how to measure the density of \eqref{aver_numb_dr_qed} in a certain spectral range.

Notice that for finite $t_{in}$, $t_{out}$ the quantity \eqref{aver_numb_dr_qed} is not zero even for a charge moving uniformly and rectilinearly. It is not surprising as, in the Schr\"{o}dinger representation, the state of the quantum electromagnetic field depends on time even for a uniformly and rectilinearly moving charge (the bound electromagnetic field depends on time at every point of space). In terms of particles, this change of the state looks as the result of creation and annihilation of photons representing the perturbations of the Fock vacuum. The same situation takes place in describing the evolution in terms of the bare photons (see the discussion in \cite{GinzbThPhAstr}).

Let us find the infrared and ultraviolet asymptotics of the expressions entering into the evolution operator for finite $t_{in}$, $t_{out}$. If
\begin{equation}
    \dot{j}^\perp_i(t,\spk)\big|_{\spk=0},\qquad\frac{\partial \dot{j}^\perp_i(t,\spk)}{\partial k_i}\Big|_{\spk=0},\quad t\in(t_{in},t_{out}),
\end{equation}
are defined, what is valid, for example, for the current $j^\perp_i(x)$ that depends smoothly on time, possesses a compact support with respect to the spatial variables, and does not have nonintegrable singularities for any $\spx$, then the term on the second line in \eqref{I_qed} is finite in the infrared region. This follows from the fact that $\dot{j}^\perp_i(t,0)\in \mathbb{R}$ and complies with the general statement made in the previous section. The quantities entering into $d(t)$ are infrared finite, too. The only singularity appears in the imaginary part of $I$,
\begin{equation}\label{I_IR}
    iI_{IR}=-\frac12\int_{IR}\frac{d\spk}{(2\pi)^3}\frac{|j^\perp_i(t,0)-j^\perp_i(t_{in},0)|^2}{2|\spk|^3} =-\frac12\int_{IR}\frac{d\omega}{6\pi^2\omega}(j_i(t,0)-j_i(t_{in},0))^2,
\end{equation}
where it is supposed that $|\spk|(t_{out}-t_{in})\ll1$. The expression \eqref{I_IR} diverges logarithmically. The zero mode has the form
\begin{equation}\label{zero_mode}
    j^i(t,0)=\int d\spx j^i(t,\spx)=\sum_n e_n\beta_n^i(t)=\frac{d}{dt}\sum_n e_nx_n^i(t)=\frac{d}{dt}d^i(t),
\end{equation}
where $d^i(t)$ is the dipole moment of the system (not to be confused with $d(t)$). This zero mode determines the leading contribution to the multipole expansion of the electromagnetic potential of a neutral system of charges at large distances from the source (see Sec. 44 of \cite{LandLifshCTF} and, for example, formula (14) of \cite{KazMult}). As a rule, this quantity is negligibly small in the multipole expansion since it is of order $|\mathbf{d}|/T$ for a system of charges evolving in a bounded domain after averaging over the interval of time $T\rightarrow\infty$. The infrared divergence appearing in \eqref{I_IR} is responsible for reconstruction of the Biot-Savart field at large distances from the nonstationary source. As it was mentioned, from physical point of view, this infrared divergence is not a problem since one can always suppose that the system under study is confined into a sufficiently large box. Furthermore, the assumption that the initial state is the ground state of the Hamiltonian of the theory is valid only in the bounded region of space. The typical size of this region or of the box can be taken as the natural infrared cutoff. Nevertheless, if the size of the chamber where the experiment is carried out is sufficiently large and $|\spk|(t_{out}-t_{in})\ll1$, then there exists a region of photon energies where the infrared asymptotics of the density of radiated photons following from \eqref{I_IR} can be observed experimentally.

As far as the ultraviolet asymptotics is concerned, the Fourier transform of an infinitely smooth current $j^\mu(t,\spx)$ tending to zero at $|\spx|\rightarrow\infty$ faster than any power of $|\spx|^{-1}$ vanishes at $|\spk|\rightarrow\infty$ faster than any power of $|\spk|^{-1}$. Therefore, all the integrals appearing in the evolution operator converge at large momenta in no-regularization  limit\footnote{Notice that such a situation does not always take place. Namely, infinitely smooth background fields of a general form rapidly vanishing at infinity, or with a compact support, may lead to the average number of particles divergent at large energies (see Introduction).}, except, of course, the energy of zero point fluctuations. Nevertheless, it is interesting to find the ultraviolet asymptotics of the average number of dressed photons created from the vacuum by the system of charged point particles \eqref{curr_part}. For such a system
\begin{equation}
    \dot{j}^i(t,\spk)=\sum_n e_n(\dot{\beta}_n^i +i\beta_n^i(\boldsymbol{\beta}_n\spk))e^{i\spk\spx_n(t)}.
\end{equation}
Substituting this expression into \eqref{aver_numb_dr_qed} in the regularization limit and integrating by parts, we find in the leading order
\begin{equation}\label{N_UV_dr0}
    N_{UV}=\int_{UV}\frac{d\spk}{(2\pi)^3}\Big[ \sum_n e^2_n\frac{(\mathbf{n}\boldsymbol{\beta}_n)^2}{2|\spk|^3} \frac{\boldsymbol{\beta}_n^2-(\mathbf{n}\boldsymbol{\beta}_n)^2}{(1-(\mathbf{n}\boldsymbol{\beta}_n))^2}\Big|_{t=t_{in}} +\sum_n e^2_n\frac{(\mathbf{n}\boldsymbol{\beta}_n)^2}{2|\spk|^3}\frac{\boldsymbol{\beta}_n^2-(\mathbf{n}\boldsymbol{\beta}_n)^2}{(1-(\mathbf{n}\boldsymbol{\beta}_n))^2}\Big|_{t=t_{out}} \Big],
\end{equation}
where $\mathbf{n}:=\spk/|\spk|$. Integrating over the angular variables, we obtain
\begin{equation}\label{N_UV_dr}
    N_{UV}=\int_{UV}\frac{d\omega}{16\pi^3\omega} \big[ \sum_n e^2_nf(\beta_n)\big|_{t=t_{in}} +\sum_n e^2_nf(\beta_n)\big|_{t=t_{out}} \big],
\end{equation}
where
\begin{equation}
    f(\beta_n)=8\pi \Big[(2-\beta_n^2)\frac{\arcth\beta_n}{\beta_n} +\frac{\beta_n^2}{3}-2 \Big]=\frac{8\pi}{15}\be^4_n+\cdots=2\pi \Big[ -\ln\frac{(1-\beta_n)^2}{4}-\frac{20}{3}+\cdots\Big].
\end{equation}
The number of particles \eqref{N_UV_dr} diverges logarithmically. In this case, the natural ultraviolet cutoff parameter is the inverse of the wave packet size. It is also clear that \eqref{N_UV_dr} does not take into account the quantum recoil due to radiation of hard photons \cite{BaiKatFad}. The account for quantum recoil results in that $|j^i(t,\spk)|$ rapidly tends to zero for $|\spk|$ larger than the total energy of the radiating particle. The current \eqref{curr_part} does not satisfy this property. To observe the asymptotics \eqref{N_UV_dr}, it is necessary that the photon energy be much smaller than $1/\tau_s$ and the ultraviolet cutoff parameter and be much larger than any typical energy of the radiation formed.

Let us find the estimate for the number of dressed photons produced during the adiabatic change of the current $j^\perp_i(t,\spk)$. Suppose that
\begin{equation}
    j^\perp_i(t,\spk)=:J^i(t/\tau,\spk),\qquad(t_{out}-t_{in})/\tau\gg1,
\end{equation}
where $\tau$ is the adiabaticity parameter and $\dot{j}^\perp_i(t,\spk)$, $\ddot{j}^\perp_i(t,\spk)$ are assumed to vanish sufficiently fast at $|\spk|\rightarrow\infty$. Let $\la$ be the infrared cutoff (see above) and $\La_{IR}(t_{out}-t_{in})\ll1$, $\La_{IR}>\la$. Then
\begin{equation}
    N=\int_\la^{\La_{IR}}\frac{d\spk}{16\pi^3|\spk|^3}\Big|\int_{t_{in}}^{t_{out}}dt e^{-i|\spk|t}\dot{j}^\perp_i(t,\spk) \Big|^2 +\int_{\La_{IR}}^\infty\frac{d\spk}{16\pi^3|\spk|^3}\Big|\int_{t_{in}}^{t_{out}}\frac{dt}{\tau} e^{-i|\spk|t} J'_i(t/\tau,\spk) \Big|^2.
\end{equation}
If $\ddot{j}^\perp_i(t,\spk)$ is absolutely integrable for $t\in[t_{in},t_{out}]$ then, on integrating by parts and using the Riemann-Lebesgue lemma, it is easy to see that the second term is of order $1/\tau^2$. This estimate is valid for $\La_{IR}\tau\gg1$. Substituting \eqref{I_IR}, \eqref{zero_mode} into the first integral, we obtain in the leading order
\begin{equation}
    N\approx\frac{\ln(\La_{IR}/\la)}{6\pi^2} (\dot{d}^i(t_{out})-\dot{d}^i(t_{in}))^2 +\int_{\La_{IR}}^\infty\frac{d\spk}{16\pi^3 \spk^4}\Big| e^{-i|\spk|t} \dot{j}^\perp_i(t,\spk)\big|^{t_{out}}_{t_{in}} \Big|^2.
\end{equation}
If $\dot{d}^i(t)$ is of order $1/\tau$, then the first term is of order $1/\tau^2$. Thus, in the adiabatic limit,
\begin{equation}\label{adiabat_dr}
    N=O(\tau^{-2}).
\end{equation}
The asymptotics \eqref{adiabat_dr} is in agreement with the standard estimate following from the uniform adiabatic theorem \cite{LandLifQED,ElgHag}. The infrared cutoff $\la$ provides the energy gap between the vacuum and the first excited state of the system.

For comparison we present here the analogous formulas for the bare photons. The general formulas \eqref{bare_part}, \eqref{bare_part1} are written as
\begin{equation}
\begin{gathered}
    f^b_\al(t)=i\frac{\bar{e}_i^{(s)}\bar{j}^i(t,\spk)}{\sqrt{2V}|\spk|^{1/2}},\qquad
    g^b_\al(t)=\int_{t_{in}}^t d\tau e^{-i|\spk|(t-\tau)} \frac{\bar{e}_i^{(s)}\bar{j}^i(t,\spk)}{\sqrt{2 V}|\spk|^{1/2}},\\
    d^b(t)=\int\frac{d\spk}{(2\pi)^3}\Big( V|\spk| +\frac{|j_0(t,\spk)|^2}{2\spk^2}\Big).
\end{gathered}
\end{equation}
Also
\begin{equation}
    \chi_\al^b(t)=\int_{t_{in}}^t d\tau e^{-i|\spk|(t_{in}-\tau)} \frac{\bar{e}_i^{(s)}\bar{j}^i(t,\spk)}{\sqrt{2 V}|\spk|^{1/2}},\qquad iI^b= -\frac{i}{2}\int_{t_{in}}^{t_{out}}dx dy j^i(x)G_{ij}(x,y)j^j(y).
\end{equation}
Recall that, in the Coulomb gauge,
\begin{equation}
    G^{(+)}_{ij}(x,y)=-i\int\frac{d\spk}{(2\pi)^3}\frac{\de^\perp_{ij}}{2|\spk|}e^{ik_\mu (x^\mu-y^\mu)}\Big|_{k^0=|\spk|}.
\end{equation}
Substituting these expressions into \eqref{numb_part_dr}, \eqref{inclus_process_prob}, \eqref{evol_bare}, we deduce the average number of bare photons recorded by the detector, the probability of the inclusive process \eqref{inclus_process}, and the matrix element of the evolution operator. The average number of produced bare photons reads as
\begin{equation}\label{aver_numb_br_qed}
    N^b=\int\frac{d\spk}{16\pi^3|\spk|}\Big|\int_{t_{in}}^{t_{out}}dt e^{-i|\spk| t} j^\perp_i(t,\spk)\Big|^2.
\end{equation}
For finite $t_{in}$, $t_{out}$, this quantity determines the average number of bare photons in the following experiment. For $t<t_{in}$ the current is shielded, the bare photons are absent in the initial state. For $t\in(t_{in},t_{out})$ the shielding is switched off. At $t=t_{out}$ the detector counts the number of bare photons and is turned off or the current is shielded once again. Such a situation can be realized, for example, by using the conducting screens: the charges are injected in the region of space where the detector is located and the shielding is absent. Then the charges escape this region and move behind the screen.

The integrals defining $d^b(t)$, $I^b$, and $N^b$ are finite in the infrared domain for $|\spk|(t_{out}-t_{in})\ll1$ provided that $j^\perp_i(t,\spk=0)$ and $j_0(t,\spk=0)$ exist. By the same reasons as in the case of dressed photons, these integrals are finite, save the energy of zero point fluctuations, in the ultraviolet region for a smooth current $j^\mu(t,\spx)$ tending to zero at $|\spx|\rightarrow\infty$ faster than any power of $|\spx|^{-1}$. As for the current of charged point particles \eqref{curr_part}, the ultraviolet asymptotics of \eqref{aver_numb_br_qed} has the form
\begin{equation}\label{N_UV_br0}
    N^b_{UV}=\int_{UV}\frac{d\spk}{16\pi^3|\spk|^3}\Big[ \sum_n e^2_n \frac{\boldsymbol{\beta}_n^2-(\mathbf{n}\boldsymbol{\beta}_n)^2}{(1-(\mathbf{n}\boldsymbol{\beta}_n))^2}\Big|_{t=t_{in}} +\sum_n e^2_n\frac{\boldsymbol{\beta}_n^2-(\mathbf{n}\boldsymbol{\beta}_n)^2}{(1-(\mathbf{n}\boldsymbol{\beta}_n))^2}\Big|_{t=t_{out}} \Big].
\end{equation}
Integrating over the angular variables, we have
\begin{equation}\label{N_UV_br}
    N_{UV}=\int_{UV}\frac{d\omega}{16\pi^3\omega} \big[ \sum_n e^2_nf^b(\beta_n)\big|_{t=t_{in}} +\sum_n e^2_nf^b(\beta_n)\big|_{t=t_{out}} \big],
\end{equation}
where
\begin{equation}
    f^b(\beta_n)=8\pi \Big(\frac{\arcth\beta_n}{\beta_n} -1 \Big)=\frac{8\pi}{3}\be^2_n+\cdots=2\pi \Big[ -\ln\frac{(1-\beta_n)^2}{4}-4+\cdots\Big].
\end{equation}
The average number of particles \eqref{N_UV_br} diverges logarithmically. It is clear from \eqref{N_UV_dr0} and \eqref{N_UV_br0} that $N_{UV}<N_{UV}^b$. This property is in accord with the general statement that the representation of the algebra of observables by means of the Hamiltonian diagonalization procedure improves the ultraviolet behavior of the theory. In the present case, the estimate \eqref{bare_dress} is not fulfilled as the estimate \eqref{current_tderiv} does not hold for the current of point particles. Notice that $f(\beta_n)\rightarrow f^b(\beta_n)$ for $\beta_n\rightarrow1$.

To conclude this section, we find the average number of bare photons created from the vacuum during the adiabatic evolution of the current $j^\perp_i(t,\spk)$. To this end, we integrate by parts with respect to $t$ in \eqref{aver_numb_br_qed}. Then, in the leading order in $1/\tau$, we obtain
\begin{equation}\label{aver_numb_br_qed_ad}
    N^b\approx\int\frac{d\spk}{16\pi^3|\spk|^3}\Big|(e^{-i|\spk|t}-1)j^\perp_i(t,\spk)\big|^{t_{out}}_{t_{in}}\Big|^2.
\end{equation}
The integral discarded in this expression tends to zero as $\tau\rightarrow\infty$ provided that $\dot{j}^\perp_i(t,\spk)$ is absolutely integrable for $t\in[t_{in},t_{out}]$ and the other assumptions about $j^\perp_i(t,\spk)$ are satisfied (see above). The quantity \eqref{aver_numb_br_qed_ad} does not tend to zero for $\tau\rightarrow\infty$. This is, of course, an expected result.

\section{Conclusion}\label{Concl}

Let us summarize the results. We developed the quantum theory of fields with nonstationary quadratic Hamiltonians of a general form for both bosons and fermions. A special attention was paid to the existence of unitary evolution during a finite interval of time in the separable Hilbert space of quantum states. To this end, the representation of the algebra of observables was realized by means of the Hamiltonian diagonalization procedure, the energy cutoff regularization was explicitly introduced into the Hamiltonian, and the divergencies in the average number of created particles were regulated by the corresponding counterdiabatic terms in the Hamiltonian. The regularized Hamiltonian is self-adjoint, local in time, and reduces to the initial Hamiltonian after the removal of regularization.

In no-regularization limit, the theory may become nonunitary due to the divergent number of created particles in the ultraviolet and/or infrared spectral domains. Nevertheless, we investigated the observables that allow for no-regularization limit. Namely, we considered the probability that the detector records a particle in a certain set of states, i.e., the probability $w(D)$ of the inclusive process \eqref{inclus_process}. In addition, we considered the average number of particles $N_D$ recorded by the detector in the aforementioned set of states. It is these quantities that are measured in experiments. We showed that under rather mild assumptions these quantities allow for the removal of regularization. In this limit, $N_D$ is finite and $w(D)\in[0,1)$ as for the regularized theory. The explicit formulas for $N_D$ and $w(D)$ were found. The formula for $w(D)$ generalizes the formula obtained in \cite{BKL5}. Of course, the issues with unitarity of the theory after the removal of regularization do not vanish. They reappear in the form of dangerous asymptotics of the average number of created particles in the ultraviolet and/or infrared regions. However, one cannot prove experimentally such violation of unitary so long as one cannot measure the number of particles at infinite and/or zero energies.

As a simple example for application of the developed formalism, we considered the theory of a neutral boson field with the Hamiltonian possessing a stationary quadratic part and a nonstationary source. We particularized the general formulas for this simple case and found the infrared and ultraviolet asymptotics of the average number of particles created from the vacuum during a finite time evolution. For such simple theories, it is not difficult to compare the observables calculated in the different representation of their algebra. Thus we found the average number of particles recorded by the detector when the Fock space is defined by means of diagonalization of the Hamiltonian without the nonstationary source (the bare particles). As a rule, for massless particles without the infrared cutoff, this representation is not unitary equivalent to the representation defined by means of diagonalization of the full nonstationary quadratic Hamiltonian (the dressed particles). We showed that the average number of dressed particles created from the vacuum possesses a better ultraviolet behavior than the same quantity for the bare particles. The infrared issues with unitarity can be resolved, for example, by placing the system into a large box. In fact, such a ``box'' is always present in any experimental setup. Then the both representations become unitary equivalent and, to a large extent, the use of different definitions of particles becomes a question of terminology. All the observables in one representation can be rewritten in the other one, although one representation can be more suitable than another for solving a given problem\footnote{Of course, there are certain restrictions on the choice of the representations of the algebra of observables. There must exist at least one representation among unitary equivalent ones that is determined by the state of the background fields at the present moment at every instant of time, i.e., in this representation, the creation-annihilation operators are the functionals of the background field configurations $\Phi(t)$. For example, such functionals can simply be independent of $\Phi(t)$. If there are not such local in time representations of the algebra of observables, one cannot pose the Cauchy problem.}. So this model is not a quite good representative for displaying the peculiarities stemming from unitary inequivalent representations of the algebra of observables. The issues with unitarity of QFT resulting from a poor ultraviolet behavior of the average number of created particles for infinitely smooth background fields with compact support are severer and cannot be resolved by analogous simple physical arguments (see the examples in \cite{uqem,Szpak15,Scharf79,NencSchr,Ruijsen}). There is not a natural ultraviolet cutoff in these model. So it has to be introduced by hand or other representations of the algebra of observables have to be considered.

Having investigated this model, we considered its particular case -- QED with a classical current in the Minkowski spacetime in the inertial reference frame. This is the classical example that was investigated in many papers and books \cite{Shirok,BlNord37,Glaub51,SchwinS,JauRohr,BaiKatFad,GinzbThPhAstr,GavrGit90}. We found the average number of dressed and bare photons created from the vacuum during a finite time evolution and the probability of the inclusive process \eqref{inclus_process}. The production of photons during the adiabatic change of the source was also studied. The infrared asymptotics of the average number of dressed and bare photons were obtained. As for the ultraviolet asymptotics, they were derived for the current of charged point particles. All these asymptotics can be verified experimentally.

As regards the possible applications of the developed formalism, they are numerous. One may mention the nonstationary problems in condensed matter physics, in QED in continuous anisotropic media and in strong electromagnetic fields, in QFTs on gravitational backgrounds, etc. It can also be used for description of finite time quantum-field processes with wave packets.

\paragraph{Acknowledgments.}

I am grateful to the anonymous referee for useful suggestions. The reported study was supported by the Russian Ministry of Education and Science, the contract N 0721-2020-0033.

\appendix
\section{Symbol of the evolution operator}\label{Symb_Evol_Oper_App}

Let us find the normal symbol of the evolution operator of QFT of bosons or fermions of a general form with nonstationary quadratic Hamiltonian. As for quadratic theories of bosons, the explicit expression for this symbol was found in \cite{KazMil1}, where, in fact, the results of \cite{BerezMSQ1.4} were generalized to a nonstationary case (see also \cite{Friedrichs} and for relatively recent studies \cite{NeretinB,BrunDerez}). Other representations of the solution to the quantum-field Schr\"{o}dinger equation with quadratic Hamiltonian are given in \cite{MaslShved,Scharf86}.

Let $(\ha_\al,\had_\al)$ be a complete set of bosonic ($\epsilon=1$) or fermionic ($\epsilon=-1$) creation-annihilation operators. By the standard means we construct the Bargmann-Fock representation \cite{Bargm61,BerezMSQ1.4,Friedrichs}. Introduce the coherent states
\begin{equation}\label{coher_stt}
\begin{gathered}
    |a\ran:=e^{\had a}|0\ran,\qquad \lan\bar{a}|:=\lan0|e^{\bar{a}\ha},\qquad \lan\bar{a}|a\ran=e^{\bar{a}a},\\
    \ha_\al|a\ran=a_\al|a\ran,\qquad \lan\bar{a}|\had_\al=\lan\bar{a}|\bar{a}_\al,
\end{gathered}
\end{equation}
where $|0\ran$ is the Fock vacuum and $a_\al$, $\bar{a}_\al$ are some functions with the Grassmann parity $(1-\epsilon)/2$. Recall that we use the matrix notation of the form \eqref{matrix_nottn}. The completeness relation reads as
\begin{equation}
    \hat{1}=\int D\bar{a}Da e^{-\bar{a}a}|a\ran\lan\bar{a}|,\qquad \int D\bar{a}Da e^{-\bar{a}a}=1.
\end{equation}
The last equality specifies the normalization of the measure of the Gaussian functional integral. This functional integral obeys the relations \cite{BerezMSQ1.4}
\begin{equation}\label{func_int_G}
\begin{split}
    \int D\bar{a}Da\exp\Big\{-\frac{1}{2}
    \left[
      \begin{array}{cc}
        a & \bar{a} \\
      \end{array}
    \right]
    B
    \left[
      \begin{array}{c}
        a \\
        \bar{a} \\
      \end{array}
    \right]
    +
     \left[
      \begin{array}{cc}
        a & \bar{a} \\
      \end{array}
    \right] F
    \Big\}&=
    \exp\big\{\tfrac{1}{2}F^T B^{-1}F\big\}
    \Big(\det\left[
      \begin{array}{cc}
        A_{21} & A_{22} \\
        A_{11} & A_{12} \\
      \end{array}
    \right]\Big)^{-1/2},\\
    \int D\bar{a}Da\exp\Big\{\frac{1}{2}
    \left[
      \begin{array}{cc}
        a & \bar{a} \\
      \end{array}
    \right]
    B
    \left[
      \begin{array}{c}
        a \\
        \bar{a} \\
      \end{array}
    \right]
    +
     \left[
      \begin{array}{cc}
        a & \bar{a} \\
      \end{array}
    \right] F
    \Big\}&=
    \exp\big\{\tfrac{1}{2}F^T B^{-1}F\big\}
    (\det B)^{1/2},
\end{split}
\end{equation}
where the first equality is for bosons, whereas the second one is for fermions. The Grassmann parity of $F$ is equal to $(1-\epsilon)/2$. Besides,
\begin{equation}
    B:=\left[
      \begin{array}{cc}
        A_{11} & A_{12} \\
        A_{21} & A_{22} \\
      \end{array}
    \right].
\end{equation}
The determinant on the first line on the right-hand side of \eqref{func_int_G} is well-defined provided $A_{11}$, $A_{22}$ are the Hilbert-Schmidt (HS) operators and $A_{12}-1$ and $A_{21}-1$ are trace-class. As for fermions, the operators $A_{12}$, $A_{21}$ must be HS, and $A_{11}-1$ and $A_{22}-1$ must be trace-class (see for details \cite{BerezMSQ1.4}).

The states of the Fock space and the kernels of operators acting in it,
\begin{equation}
    \Phi(\bar{a}):= \lan\bar{a}|\Phi\ran,\qquad \tilde{A}(\bar{a},a):=\lan\bar{a}|\hat{A}|a\ran,
\end{equation}
are the functionals of $a_\al$, $\bar{a}_\al$. It is clear from \eqref{coher_stt} that
\begin{equation}
    \tilde{A}(\bar{a},a)=A(\bar{a},a)e^{\bar{a}a},
\end{equation}
where $A(\bar{a},a)$ is the normal (Wick) symbol of the operator $\hat{A}$. The functional
\begin{equation}
    \bar{\Phi}(a):=\lan\Phi|a\ran
\end{equation}
is obtained from $\Phi(\bar{a})$ by the complex conjugation that, in particular, replaces $\bar{a}_\al\rightarrow a_\al$ and arranges the functions $a_\al$ in the inverse order as on transposition. The same is valid for the functional corresponding to the kernel of the operator $\hat{A}^\dag$ and for the normal symbol of the operator $\hat{A}^\dag$, viz., one should take the complex conjugation and arrange the functions $a_\al$, $\bar{a}_\al$ in the inverse order. In the Bargmann-Fock representation, we have
\begin{equation}\label{rels_symb}
\begin{alignedat}{2}
    \hat{a}_\al|\Phi\ran &\leftrightarrow\frac{\de\Phi(\bar{a})}{\de\bar{a}_\al},&\qquad \hat{a}^\dag_\al|\Phi\ran &\leftrightarrow \bar{a}_\al\Phi(\bar{a}),\\
    \hat{a}_\al\hat{A} &\leftrightarrow \Big(a_\al+\frac{\de}{\de\bar{a}_\al}\Big)A(\bar{a},a),&\qquad \hat{a}^\dag_\al\hat{A}&\leftrightarrow \bar{a}_\al A(\bar{a},a),\\
    \hat{A}\hat{a}_\al &\leftrightarrow A(\bar{a},a)a_\al,&\qquad \hat{A}\hat{a}^\dag_\al &\leftrightarrow A(\bar{a},a)\Big(\bar{a}_\al+\frac{\overleftarrow{\de}}{\de a_\al}\Big),\\
    \hat{a}_\al\hat{A} &\leftrightarrow \frac{\de \tilde{A}(\bar{a},a)}{\de\bar{a}_\al},&\qquad \hat{a}^\dag_\al\hat{A} &\leftrightarrow \bar{a}_\al \tilde{A}(\bar{a},a),\\
    \hat{A}\hat{a}_\al &\leftrightarrow \tilde{A}(\bar{a},a)a_\al,&\qquad \hat{A}\hat{a}^\dag_\al &\leftrightarrow \tilde{A}(\bar{a},a)\frac{\overleftarrow{\de}}{\de a_\al},
\end{alignedat}
\end{equation}
for both bosons and fermions.

Let us given the two sets of the creation-annihilation operators $(\ha_\al,\had_\al)$ and $(\hb_\al,\hbd_\al)$ related by the linear canonical transform
\begin{equation}\label{canon_trans}
    \left[
       \begin{array}{c}
         \hat{b} \\
         \hat{b}^\dag \\
       \end{array}
     \right]=\left[
               \begin{array}{cc}
                 \Phi & \Psi \\
                 \bar{\Psi} & \bar{\Phi} \\
               \end{array}
             \right]\left[
                      \begin{array}{c}
                        \hat{a} \\
                        \hat{a}^\dag \\
                      \end{array}
                    \right]+\left[
                              \begin{array}{c}
                                f \\
                                \bar{f} \\
                              \end{array}
                            \right],
\end{equation}
where
\begin{equation}
    \left[
       \begin{array}{cc}
         \Phi & \Psi \\
         \bar{\Psi} & \bar{\Phi} \\
       \end{array}
     \right]\left[
              \begin{array}{cc}
                0 & 1 \\
                -\epsilon & 0 \\
              \end{array}
            \right]\left[
                     \begin{array}{cc}
                       \Phi^T & \Psi^\dag \\
                       \Psi^T & \Phi^\dag \\
                     \end{array}
                   \right]=\left[
              \begin{array}{cc}
                0 & 1 \\
                -\epsilon & 0 \\
              \end{array}
            \right],
\end{equation}
or
\begin{equation}\label{canon_defn}
    \left[
       \begin{array}{cc}
         \Phi & \Psi \\
         \bar{\Psi} & \bar{\Phi} \\
       \end{array}
     \right]
     \left[
       \begin{array}{cc}
         \Phi^\dag & -\epsilon\Psi^T \\
         -\epsilon\Psi^\dag & \Phi^T \\
       \end{array}
     \right]=
     \left[
       \begin{array}{cc}
         \Phi^\dag & -\epsilon\Psi^T \\
         -\epsilon\Psi^\dag & \Phi^T \\
       \end{array}
     \right]
     \left[
       \begin{array}{cc}
         \Phi & \Psi \\
         \bar{\Psi} & \bar{\Phi} \\
       \end{array}
     \right]=
     \left[
       \begin{array}{cc}
         1 & 0 \\
         0 & 1 \\
       \end{array}
     \right].
\end{equation}
In terms of components, we have
\begin{equation}\label{canon_defn1}
\begin{alignedat}{2}
    \Phi\Phi^\dag&=1+\epsilon\Psi\Psi^\dag,&\qquad\Phi^\dag\Phi&=1+\epsilon\Psi^T\bar{\Psi},\\
    \Phi\Psi^T&=\epsilon\Psi\Phi^T,&\qquad\Phi^\dag\Psi&=\epsilon\Psi^T\bar{\Phi},\\
    (\Phi\Phi^\dag)^{-1}&=1-(\Phi^\dag)^{-1}\Psi^T(\Phi^T)^{-1}\Psi^\dag,&\qquad (\Phi^\dag\Phi)^{-1}&=1-\Phi^{-1}\Psi\bar{\Phi}^{-1}\bar{\Psi},\\
    \Psi\bar{\Phi}^{-1}&=\epsilon(\Psi\bar{\Phi}^{-1})^T,&\qquad\Phi^{-1}\Psi&=\epsilon(\Phi^{-1}\Psi)^T,
\end{alignedat}
\end{equation}
where it is assumed on the last two lines that there exists the bounded operator $\Phi^{-1}$. This is always valid for bosons. As for fermions, we will suppose that $\Phi$ possesses a bounded inverse. The case of degenerate $\Phi$ with even-dimensional kernel can be obtained by a limiting process from the nondegenerate case \cite{BerezMSQ1.4}. Notice that it follows from \eqref{canon_defn1} that the operators $\Phi$ and $\Psi$ are bounded in the fermionic case.

\begin{thm}\label{canon_unit_thm}
  The linear canonical transform \eqref{canon_trans}, \eqref{canon_defn} corresponds to the unitary transform
  \begin{equation}\label{unit_trans_2}
    \hat{b}_\al=\hat{U}\hat{a}_\al\hat{U}^\dag,\qquad \hat{b}^\dag_\al=\hat{U}\hat{a}^\dag_\al\hat{U}^\dag,
  \end{equation}
  if and only if
  \begin{enumerate}
    \item $\Psi$ is HS
    \begin{equation}
        \Sp(\Psi^\dag\Psi)<\infty;
    \end{equation}
    \item $f_\al$ belongs to the Hilbert space
\begin{equation}
    \bar{f} f<\infty.
\end{equation}
  \end{enumerate}
  In this case, the matrix element $\tilde{U}(\bar{a},a)$ of the operator $\hat{U}$ takes the form
  \begin{equation}\label{unit_trans}
  \begin{split}
    \tilde{U}=&c\exp\Big\{\frac12\left[
                                  \begin{array}{cc}
                                    a & \bar{a} \\
                                  \end{array}
                                \right]
    \left[
                                             \begin{array}{cc}
                                               \epsilon\bar{\Psi}\Phi^{-1} & \epsilon(\Phi^{-1})^T \\
                                               \Phi^{-1} & -\Phi^{-1}\Psi \\
                                             \end{array}
                                           \right]\left[
                                                    \begin{array}{c}
                                                      a \\
                                                      \bar{a} \\
                                                    \end{array}
                                                  \right]+a(\bar{f}-\bar{\Psi}\Phi^{-1}f)-\bar{a}\Phi^{-1}f
     \Big\},\\
     c=&\frac{e^{i\vf}}{(\det\Phi\Phi^\dag)^{\epsilon/4}}\exp\Big\{\frac14\left[
                                                                     \begin{array}{cc}
                                                                       f & \bar{f} \\
                                                                     \end{array}
                                                                   \right]
                                                                   \left[
                                                                     \begin{array}{cc}
                                                                       (\Phi^{-1})^T\Psi^\dag & -\epsilon \\
                                                                       -1 &  \epsilon(\Phi^{-1})^\dag\Psi^T\\
                                                                     \end{array}
                                                                   \right]
                                                                   \left[
                                                                     \begin{array}{c}
                                                                       f \\
                                                                       \bar{f} \\
                                                                     \end{array}
                                                                   \right]
      \Big\},  \end{split}
  \end{equation}
  where $\vf$ is an arbitrary phase.
\end{thm}

The proof of this theorem is given in \cite{BerezMSQ1.4}. As long as the operator $\Psi$ is HS in the theorem, $\Phi\Phi^\dag-1$ is trace-class and the Fredholm determinant in \eqref{unit_trans} is well-defined. This determinant is not zero for fermions inasmuch as we assume that $\Phi$ is nondegenerate.

Let the Hamiltonian of the system be
\begin{equation}\label{Hamil_gener}
    \hat{H}=\frac12\big[2\hat{a}^\dag C(t)\hat{a}+\hat{a} A^\dag(t) \hat{a} +\hat{a}^\dag A(t) \hat{a}^\dag\big]+\hat{a}^\dag f(t)+\bar{f}(t)\hat{a}+d(t),
\end{equation}
where $C(t)=C^\dag(t)$ is a self-adjoint operator, $A^T(t)=\epsilon A(t)$ is an (anti)symmetric operator, $f_\al(t)$ are some functions of the Grassmann parity $(1-\epsilon)/2$, and $d(t)$ is a Grassmann even function. Introduce the standard notation for the operator norms
\begin{equation}
    \|A\|_1:=\Sp\sqrt{A^\dag A},\qquad \|A\|_2:=\big[\Sp(A^\dag A)]^{1/2}.
\end{equation}

\begin{thm}\label{evol_symb_thm}
Let the Hamiltonian of the system take the form \eqref{Hamil_gener} and the following conditions be satisfied:
\begin{enumerate}
    \item There exists the unitary operator
    \begin{equation}\label{V_oper}
        R_{\tau,0}=\Texp\Big[-i\int_0^\tau ds C(s)\Big],\qquad\tau\in[0,t];
    \end{equation}
    \item The operator $A(\tau)$ is uniformly bounded for $\tau\in[0,t]$:
    \begin{equation}
        \exists a>0:\|A(\tau)\|<a,\forall\tau\in[0,t];
    \end{equation}
    \item The operator
    \begin{equation}
        F(\tau):=\int_0^{\tau}ds \bar{R}_{0,s}\bar{A}(s) R_{s,0}
    \end{equation}
    is HS and $\|F(\tau)\|_2$ is locally integrable for $\tau\in[0,t]$;
    \item The operator
    \begin{equation}
        G(\tau):=R_{0,\tau}A(\tau)\bar{R}_{\tau,0}F(\tau)
    \end{equation}
    is trace-class and $\|G(\tau)\|_1$ is locally integrable for $\tau\in[0,t]$;
    \item $\bar{f}(\tau)f(\tau)<b<\infty$ and $d(\tau)$ is absolutely locally integrable for $\tau\in[0,t]$;
    \item In the fermionic case, the operator $\Phi(\tau)$ defined in \eqref{unit_canon} has a bounded inverse for $\tau\in[0,t]$.
\end{enumerate}
Then the matrix element of the evolution operator $\hat{U}_{t,0}$ is written as
\begin{equation}\label{evol_symb}
    \tilde{U}_{t,0}(\bar{a},a)=c(t)\exp\Big\{\frac12\left[
                                                      \begin{array}{cc}
                                                        \bar{a} & a \\
                                                      \end{array}
                                                    \right]
                                                    \left[
                                                      \begin{array}{cc}
                                                        \Psi\bar{\Phi}^{-1} & (\Phi^\dag)^{-1} \\
                                                        \epsilon\bar{\Phi}^{-1} & -\epsilon\bar{\Phi}^{-1}\bar{\Psi} \\
                                                      \end{array}
                                                    \right]
                                                    \left[
                                                      \begin{array}{c}
                                                        \bar{a} \\
                                                        a \\
                                                      \end{array}
                                                    \right]
                                                    +\bar{a}(\Phi^\dag)^{-1}\chi-\epsilon a(\bar{\chi}+\bar{\Phi}^{-1}\bar{\Psi}\chi)
    \Big\},
\end{equation}
where
\begin{equation}\label{gandc}
  \begin{split}
    \left[
      \begin{array}{c}
        \chi \\
        \bar{\chi} \\
      \end{array}
    \right]&=-i\int_0^t d\tau D_{0,\tau}
                                         \left[
                                           \begin{array}{c}
                                             f(\tau) \\
                                             -\bar{f}(\tau) \\
                                           \end{array}
                                         \right],\\
    c(t)&=\big[\det \bar{R}_{0,t}\bar{\Phi}(t)\big]^{-\epsilon/2}\exp\Big\{-i\int_0^td\tau \Big[\frac{\epsilon}2\chi\bar{\Phi}^{-1}\bar{A}(\Phi^\dag)^{-1}\chi+\bar{f}(\Phi^\dag)^{-1}\chi+d(\tau)\Big] \Big\},
  \end{split}
\end{equation}
and
\begin{equation}\label{unit_canon}
    D_{t,0}=\left[
      \begin{array}{cc}
        \Phi(t) & \Psi(t) \\
        \bar{\Psi}(t) & \bar{\Phi}(t) \\
      \end{array}
    \right]
    =\Texp\Big\{-i\int_0^t d\tau \left[
                                           \begin{array}{cc}
                                             C(\tau) & A(\tau) \\
                                             -\bar{A}(\tau) & -\bar{C}(\tau) \\
                                           \end{array}
                                         \right]\Big\}.
\end{equation}
\end{thm}
\begin{proof}
The proof of this statement is the same as given in \cite{BerezMSQ1.4} with the exception that now the Hamiltonian depends on time. The formal proof of this theorem for bosons is presented in \cite{KazMil1}. As for fermions, the formal proof is conducted along the same lines as for bosons apart from some signs arising due to anticommutativity of $a_\al$, $\bar{a}_\al$, $f_\al$, and $\bar{f}_\al$.

The proof of the existence of \eqref{evol_symb} is reduced to the proof of the existence of \eqref{unit_canon} and that the operator $\Psi$ is HS, the operator $\bar{R}_{0,t}\bar{\Phi}(t)-1$ is trace-class, and the expressions in the exponents \eqref{evol_symb}, \eqref{gandc} are bounded. Let
\begin{equation}
    h(t):=\left[
       \begin{array}{cc}
         C(t) & A(t) \\
         -\bar{A}(t) & -\bar{C}(t) \\
       \end{array}
     \right]=\left[
       \begin{array}{cc}
         C(t) & 0 \\
         0 & -\bar{C}(t) \\
       \end{array}
     \right]+
     \left[
       \begin{array}{cc}
         0 & A(t) \\
         -\bar{A}(t) & 0 \\
       \end{array}
     \right]=:h_0(t)+v(t),
\end{equation}
where $h_0(t)$ is the first matrix and $v(t)$ is the second one. Introduce the operator
\begin{equation}\label{det_pert_theor}
    S_{t,0}:=U^0_{0,t}D_{t,0}=
    \left[
      \begin{array}{cc}
        L(t) & M(t) \\
        \bar{M}(t) & \bar{L}(t) \\
      \end{array}
    \right],
    \qquad
    U^0_{t_2,t_1}=
    \left[
      \begin{array}{cc}
        R_{t_2,t_1} & 0 \\
        0 & \bar{R}_{t_2,t_1} \\
      \end{array}
    \right].
\end{equation}
There is the standard representation for the operator $S_{t,0}$ in the form of the series of nonstationary perturbation theory
\begin{equation}
    S_{t,0}=1+\sum_{n=1}^\infty S^{(n)}_{t,0},\qquad S^{(n)}_{t,0}=-i\int_0^t d\tau v_I(\tau)S^{(n-1)}_{\tau,0},
\end{equation}
and
\begin{equation}
    v_I(t)=U^0_{0,t}v(t)U^0_{t,0}.
\end{equation}
Substituting the matrix representations for the operators, we come to the recurrence relations
\begin{equation}\label{recurr_rels}
\begin{split}
    \bar{M}^{(n)}(t)&=\int_0^tdt_1\int_0^{t_1}dt_2\bar{R}_{0,t_1}\bar{A}(t_1)R_{t_1,0}R_{0,t_2}A(t_2)\bar{R}_{t_2,0}\bar{M}^{(n-2)}(t_2),\\
    L^{(n)}(t)&=\int_0^tdt_1\int_0^{t_1}dt_2R_{0,t_1}A(t_1)\bar{R}_{t_1,0}\bar{R}_{0,t_2}\bar{A}(t_2)R_{t_2,0}L^{(n-2)}(t_2),
\end{split}
\end{equation}
and
\begin{equation}\label{low_orders}
\begin{gathered}
    L^{(0)}(t)=1,\qquad \bar{M}^{(1)}(t)=iF(t),\qquad L^{(2)}(t)=\int_0^t d\tau G(\tau),\\
    \bar{M}^{(2n)}=L^{(2n+1)}=0,\; n=\overline{0,\infty}.
\end{gathered}
\end{equation}
It follows from the properties of the operator norms (see, e.g., \cite{RSv1,Shubin}) and the recurrence relations \eqref{recurr_rels} that
\begin{equation}
\begin{split}
    \|\bar{M}^{(n)}(t)\|_2& \leqslant a^2 \int_0^tdt_1\int_0^{t_1}dt_2\|\bar{M}^{(n-2)}(t_2)\|_2,\\
    \|L^{(n)}(t)\|_1& \leqslant a^2 \int_0^tdt_1\int_0^{t_1}dt_2 \|L^{(n-2)}(t_2)\|_1.
\end{split}
\end{equation}
Using these recurrence relations and the initial data \eqref{low_orders}, it is easy to see that on fulfillment of conditions of the theorem
\begin{equation}\label{estimates}
    \Big\|\sum_{n=0}^\infty \bar{M}^{(n)}(t)\Big\|_2<\infty,\qquad \Big\|\sum_{n=1}^\infty L^{(n)}(t)\Big\|_1<\infty.
\end{equation}
The first inequality implies that the operator
\begin{equation}
    \bar{M}(t)=\bar{R}_{0,t}\bar{\Psi}(t)
\end{equation}
is HS. Consequently, $\bar{\Psi}(t)$ is also HS. The second inequality in \eqref{estimates} implies that the operator
\begin{equation}
    L(t)-1=R_{0,t}\Phi(t)-1
\end{equation}
is trace-class. This proves the existence of the determinant entering into \eqref{gandc}, the existence of the linear canonical transform \eqref{unit_canon}, and that the latter corresponds to the unitary transform. The boundedness of the expressions in the exponents in \eqref{evol_symb}, \eqref{gandc} is evident under the assumptions of the theorem.

\end{proof}

Some assumptions of the theorem can be relaxed but we will not investigate this point here. In particular, for fermions, in the case when $\Phi(t)$ possesses an even-dimensional kernel, the matrix elements of the evolution operator can be obtained from \eqref{evol_symb} by a passage to the limit in the formula for the nondegenerate case \cite{BerezMSQ1.4}. A thorough investigation of the dynamics of electrons in the overcritical fields resulting in degeneracy of the operator $\Phi(t)$ can be found in \cite{Szpak15,GrMuRaf}.

We shall also need the relation between the creation-annihilation operators in the Heisenberg representation
\begin{equation}
    \hat{a}_\al(t):=\hat{U}_{0,t}\hat{a}_\al(0)\hat{U}_{t,0},\qquad \hat{a}^\dag_\al(t):=\hat{U}_{0,t}\hat{a}^\dag_\al(0)\hat{U}_{t,0}.
\end{equation}
These operators obey the Heisenberg equations with the Hamiltonian \eqref{Hamil_gener},
\begin{equation}
    i\left[
       \begin{array}{c}
         \dot{\hat{a}}(t) \\
         \dot{\hat{a}}^\dag(t) \\
       \end{array}
     \right]=
     h(t)
     \left[
       \begin{array}{c}
         \hat{a}(t) \\
         \hat{a}^\dag(t) \\
       \end{array}
     \right]+
     \left[
       \begin{array}{c}
         f(t) \\
         -\bar{f}(t) \\
       \end{array}
     \right].
\end{equation}
The solution of these equations is
\begin{equation}\label{unit_trans_3}
    \left[
      \begin{array}{c}
        \hat{a}(t) \\
        \hat{a}^\dag(t) \\
      \end{array}
    \right]=
    D_{t,0}
    \left[
      \begin{array}{c}
        \hat{a}(0) \\
        \hat{a}^\dag(0) \\
      \end{array}
    \right]
    +\left[
      \begin{array}{c}
        g(t) \\
        \bar{g}(t) \\
      \end{array}
    \right]=D_{t,0}
    \left[
      \begin{array}{c}
        \hat{a}(0) \\
        \hat{a}^\dag(0) \\
      \end{array}
    \right]
    +D_{t,0}
    \left[
      \begin{array}{c}
        \chi(t) \\
        \bar{\chi}(t) \\
      \end{array}
    \right],
\end{equation}
where $D_{t,0}$ and $\chi(t)$ are defined in the formulation of the theorem \ref{evol_symb_thm}.

Notice that
\begin{equation}\label{vacuum-vacuum}
    \tilde{U}_{t,0}(0,0)|_{f_\al=\bar{f}_\al=0}=\big[\det \bar{R}_{0,t}\bar{\Phi}(t)\big]^{-\epsilon/2} e^{-i\int_0^t d\tau d(\tau)}.
\end{equation}
This expression determines the vacuum-to-vacuum amplitude \eqref{vacuum-vacuum_expl} in the absence of sources and, in fact, gives the one-loop effective action of the theory \cite{DeWGAQFT.11,Schwing.det,Schwing.10,BirDav.11,GrMuRaf,FullingAQFT,GFSh.3,BuchOdinShap.11,GriMaMos.11,DeWGAQFT.11,CalzHu,ParkTom,WeinB2}. The function $c(t)$ in the presence of sources is the unnormalized generating functional of free Green's functions, i.e., the Green's functions of quadratic theory on the given classical background. The expression \eqref{vacuum-vacuum} can be rewritten in other form under the additional assumption that $C(\tau)$ in \eqref{V_oper} is trace-class and $\|C(\tau)\|_1$ is locally integrable for $\tau\in[0,t]$. Then using the nonstationary perturbation theory as in the proof of the theorem \ref{evol_symb_thm}, it is not difficult to show that $R_{t,0}-1$ is trace-class and continuously depends on $C(\tau)$, $\tau\in[0,t]$, with respect to the norm $\|\cdot\|_1$. If $C(\tau)$ is a finite-rank operator for $\tau\in[0,t]$, then the Liouville theorem holds
\begin{equation}\label{Liouv_thm}
    \det R_{t,0}=\exp\Big\{-i\int_0^t d\tau\Sp C(\tau)\Big\}.
\end{equation}
The Fredholm determinant $\det(1+X)$ is a continuous function of $X$ with respect to the norm $\|\cdot\|_1$ (see, e.g., \cite{GohbGoldKrup}). The algebra of finite rank operators is a dense subset in the trace-class operators with respect to $\|\cdot\|_1$. Consequently, taking the limit in \eqref{Liouv_thm}, we see that \eqref{Liouv_thm} is valid when $C(\tau)$ is trace-class and $\|C(\tau)\|_1$ is locally integrable for $\tau\in[0,t]$. In this case, under the assumptions of the theorem \ref{evol_symb_thm}, we have
\begin{equation}\label{vacuum-vacuum_gen}
    \tilde{U}_{t,0}(0,0)|_{f_\al=\bar{f}_\al=0}=\big[\det \bar{\Phi}(t)\big]^{-\epsilon/2} e^{i\int_0^td\tau\big[\frac{\epsilon}{2}\Sp C(\tau) - d(\tau)\big]}.
\end{equation}
We shall simplify this formula further in considering the concrete models in Sec. \ref{Gener_Form}.

\end{document}